\documentclass[aps,pra,superscriptaddress,12pt,nofootinbib,longbibliography]{revtex4-2}
\usepackage{algpseudocode}
\usepackage{algorithm}
\usepackage{physics}
\usepackage{dcolumn} 
\usepackage{bm}       
\usepackage{amsmath,amssymb,amsfonts,amsthm}
\usepackage{mathtools}
\usepackage[mathscr]{eucal}
\usepackage{xcolor}
\usepackage{tikz}
\usetikzlibrary{matrix,arrows,decorations.pathmorphing}
\usepackage[utf8]{inputenc}
\usepackage{xparse}
\usepackage{hyperref}
\usepackage{cleveref}
\usepackage{listings}
\usepackage{color}

\usepackage{comment}
\usepackage{dsfont}
\usepackage{url}
\usepackage{stmaryrd} 
\usepackage{graphicx,xcolor}
\usepackage{hyperref}
\usepackage[normalem]{ulem}
\usepackage{nameref}
\usepackage[margin = 1in,footskip=.5in]{geometry}
\hypersetup{colorlinks=true,citecolor=blue,linkcolor=blue,filecolor=blue,urlcolor=blue}
\usepackage{braket}

\providecommand{\ignore}[1]{}

\newif\ifcmnt
\cmnttrue
\ifdefined\cmntsoff\cmntfalse\fi

\ifcmnt

\providecommand{\aucmnt}[1]{#1}

\else
\providecommand{\aucmnt}[1]{}

\fi

\newtheorem{thm}{Theorem}[section]
\newtheorem{prop}[thm]{Proposition}
\newtheorem{lem}[thm]{Lemma}
\newtheorem{conj}[thm]{Conjecture}
\newtheorem{cor}[thm]{Corollary}

\theoremstyle{definition}
\newtheorem{definition}[thm]{Definition}
\newtheorem{remark}[thm]{Remark}
\newtheorem{example}[thm]{Example}
\numberwithin{equation}{section}

\renewcommand\bra[1]{{\langle{#1}|}}
\renewcommand\ket[1]{{|{#1}\rangle}}
\newcommand{\id}{\text{id}}

\newcommand{\cB}{\mathcal{B}}

\newcommand{\cK}{\mathcal{K}}

\numberwithin{equation}{section}

\newtheorem*{theorem*}{Theorem}

\newtheorem*{proposition*}{Proposition}

\newtheorem*{lemma*}{Lemma}

\newcommand{\one}{\mathds{1}}

\begin{document}
\title{Towards a resolution of the spin alignment problem}

\author{Mohammad A. Alhejji\footnote{mohammad.alhejji@colorado.edu}}
\affiliation{National Institute of Standards and Technology, Boulder, Colorado 80305, USA}
\affiliation{JILA, University of Colorado, 440 UCB, Boulder, CO 80309, USA}
\affiliation{Department of Physics, University of Colorado, Boulder, Colorado, USA}
\author{Emanuel Knill}
\affiliation{National Institute of Standards and Technology, Boulder, Colorado 80305, USA}
\affiliation{Center for Theory of Quantum Matter, University of Colorado, Boulder, Colorado 80309, USA}

\begin{abstract}
  Consider minimizing the entropy of a mixture of
  states by choosing each state subject to constraints. If the
  spectrum of each state is fixed, we expect that in order to reduce the
  entropy of the mixture, we should make the states less
  distinguishable in some sense.  Here, we study a class of
  optimization problems that are inspired by
  this situation and shed light on the relevant notions of
  distinguishability. The motivation for our study is the recently introduced 
  spin alignment conjecture. In the original version of the underlying problem,
  each state in the mixture is constrained to be a freely chosen
  state on a subset of $n$ qubits tensored with a fixed state $Q$ on each of the qubits in the complement. According to the conjecture, the entropy of the mixture is minimized by choosing the freely chosen state in each term to be a tensor product of projectors onto a fixed maximal eigenvector of $Q$, which maximally ``aligns'' the terms in the mixture. We generalize this problem in several ways.  First, instead of minimizing entropy, we consider maximizing arbitrary unitarily invariant convex functions such as Fan norms and Schatten norms. To formalize and generalize the conjectured required alignment, we
  define \textit{alignment} as a preorder on tuples of self-adjoint operators that is induced by majorization. We prove the generalized conjecture for Schatten
  norms of integer order, for the case where the freely chosen
  states are constrained to be classical, and for the case where only
  two states contribute to the mixture and $Q$ is
  proportional to a projector.  The last case fits into a more
  general situation where we give explicit conditions for maximal 
  alignment. The spin alignment problem has a natural ``dual"
  formulation, versions of which have further generalizations that we
  introduce. 
\end{abstract}
\maketitle

\tableofcontents

\section{Introduction}
\label{sec: intro}

Minimizing dispersion in a communication
signal is often necessary in both the theory and practice of
information processing. Informally, the term dispersion is used here so that the more dispersion in a signal the more information content it carries.  
How dispersion is quantified depends on context. For example, in a setting where sources are taken to generate signals in an independent
and identically-distributed manner and processes are assumed to be
memoryless, a useful measure of dispersion is the Shannon entropy for classical signals associated with probability distributions
\cite{Shannon1948}, or analogously the von Neumann entropy for quantum
signals associated with quantum states
\cite{Schumacher1995}. By this point, myriads of measures of dispersion have
been introduced and extensively studied
\cite{Renyi1960, Tsallis1988, Chehade2019}. A sensible requirement for a
measure of dispersion is that it does not decrease under any process
where labels designating outcomes are confused. Such a
process is represented classically by a doubly-stochastic transition matrix; the
quantum generalization of which is a mixed-unitary channel. Put in other words, we require measures of
dispersion to be \textit{Schur-concave}.  In this work, we switch from concave to convex and primarily focus on convex Schur-convex functions. To be precise, we study how unitarily invariant convex functions behave under mixing of
signals. The value of these functions at a mixture of signals depends
on the spectra of the individual signals supported in the mixture as
well as the alignment or overlap of the signals. It is our aim to shed
light on the latter kind of dependence.

Maximizing unitarily invariant convex functions is less ubiquitous in quantum information than minimizing them, and is generally a more difficult task. This difficulty merits more analysis and research. Quantum channel minimum output entropy is an information-theoretic quantity that is defined by a minimization of a Schur-concave function. A reason it is of interest is the equivalence of the additivity (under tensor multiplication) conjectures of the minimum output entropy and the Holevo information \cite{Shor2004}. The latter quantity is relevant because its regularization is a formula for the capacity for classical communication \cite{Schumacher1997}. Despite an existence proof due to Hastings \cite{Hastings2009}, an explicit example of a quantum channel with a finite-dimensional domain and a strictly super-additive minimum output entropy is yet to be found. In addition, such minimizations are relevant in broadcast channel scenarios with privacy concerns. Ideally, the entropy of the environment or the adversary is kept at a minimum. Depending on the channels in question, the optimal states may be mixtures, as in the case of platypus channels \cite{Leditzky2022a}.

We are interested in situations where we can unambiguously say that a signal $a$ contains less dispersion than a signal $b$. That is, $f(b) \leq f(a)$ for all Schur-convex functions $f$. Luckily, we need not argue such an inequality for every such \(f\). It suffices to check that $a$ majorizes $b$. In this sense, of all the ways of measuring dispersion, majorization is the most stringent \cite{Marshall2011, Bhatia1997}. Majorization theory, which is nearly a century old \cite{Hardy1929, hardylittlewood_1934}, is the basis of Nielsen's celebrated theorem for pure bipartite state conversion using local operations and classical communication (LOCC) \cite{Nielsen1999}. This is one example among many of majorization giving insights into information processing.

Here, we elucidate what we mean by alignment and how it is relevant to the behavior of dispersion under mixing. For a self-adjoint operator \(S\), let $\lambda(S)$ denote the tuple whose entries are the eigenvalues of \(S\) (with multiplicity) ordered non-increasingly. 
\begin{definition}
\label{def: perfect alignment} (perfect alignment)
    Let $S_{1},\ldots, S_{\ell}$ be self-adjoint operators on $\mathbb{C}^{d}$. They are said to be \textit{perfectly aligned} if there exists an ordered orthonormal basis $\{ \ket{\phi_{i}} \}_{i=1}^{d}$ such that $S_{1} = \sum_{i =1}^{d} \lambda_{i} (S_{1}) \ket{\phi_{i}}\bra{\phi_{i}}, \: \ldots \: , S_{\ell} = \sum_{i =1}^{d} \lambda_{i} (S_{\ell}) \ket{\phi_{i}}\bra{\phi_{i}}$.
\end{definition}
As an example, we note that self-adjoint projection operators $P_1, \ldots, P_\ell$ are perfectly aligned if and only if their supports are nested. This definition is motivated by Fan's majorization relation \cite{Fan1949}, which in the case of states $\rho_{1},\ldots, \rho_{\ell}$ and a probability measure $p = (p_{i})_{i=1}^{\ell}$ reads
\begin{align} 
\label{Fan's majorization relation}
\lambda(\sum_{i=1}^{\ell} p_{i} \rho_{i}) \preceq \sum_{i=1}^{\ell} p_{i} \lambda (\rho_{i}).
\end{align}
Here \(b\preceq a\) means \(b\) is majorized by \(a\). These inequalities are simultaneously saturated if and only if the states are perfectly aligned. Hence, to minimize the dispersion in a mixture of states subject to spectral constraints, we should choose the states to be perfectly aligned.

More generally, we may consider a set \(G\) whose elements are \(\ell\)-tuples of states
\( (\rho_{i})_{i=1}^{\ell}\) that satisfy constraints including that for each \(i \in [\ell]\), \(\lambda(\rho_i)\) is given.  We may then ask if there is a tuple \( (\tilde{\rho}_{i})_{i=1}^{\ell} \in G\) such that for all probability measures
  \((p_{i})_{i=1}^{\ell}\),
  \(\lambda(\sum_{i=1}^{\ell} p_{i} \tilde{\rho}_{i})\) majorizes
  \(\lambda(\sum_{i=1}^{\ell} p_{i} {\sigma}_{i})\) for all \((\sigma_{i})_{i=1}^{\ell} \in G\). If \(G\) contains a tuple of perfectly aligned states, the answer is yes 
by way of the relation in Eq.~\ref{Fan's majorization relation}. The following is an example of a situation where constraints preclude the existence of such a tuple. 

\begin{example}
  \label{ex: toy example}
    Let $\tau$ be a state acting on $\mathbb{C}^d$ of rank strictly less than $d$. Let $\ket{\alpha}$ be
    a unit vector in $\text{supp}(\tau) := \text{ker}(\tau)^{\perp}$ and let $\ket{e}$ be a unit
    vector in $\text{ker}(\tau)$. Consider the set of pairs of states
    of the form $(\tau,\ketbra{v_{\gamma}})$,
    where $\ket{v_{\gamma}} := \sqrt{\gamma} \ket{\alpha} + \sqrt{1 -
      \gamma} \ket{e}$ for some $\gamma\in[0,1]$.
    If $\ket{\alpha}$ is a maximal eigenvector of $\tau$, then $\tau$ is perfectly aligned with $\ketbra{v_{1}}$. Hence, we can see via Eq.~\ref{Fan's majorization relation} that
    for all $p \in [0,1]$, 
    $p\tau+(1-p)\ketbra{v_{1}}$ majorizes $p\tau+(1-p)\ketbra{v_{\gamma}}$
    for all $\gamma\in[0,1]$.
     If $\ket{\alpha}$ is not a maximal eigenvector of $\tau$, then
    $\tau$ is not perfectly aligned with $\ketbra{v_{\gamma}}$ for any $\gamma\in[0,1]$.
    However, it is natural to conjecture that the choice $\gamma = 1$ minimizes dispersion in the mixture anyway,
    as that is where the two states are as overlapped or aligned as possible.
    We give a proof of this conjecture in App.~\ref{sec:app C}.
\end{example}

We use the spin alignment problem as a case study of these notions of alignment with the aim of refining Fan's majorization relation. Introduced by Leditzky et al. in Ref.~\cite{Leditzky2022a}, the problem arose in the context of proving additivity of the coherent information of a class of quantum channels with peculiar communication properties. For a qudit state $Q$ and a probability measure $\mu$, they considered $n$-qudit states of the form
\begin{align} 
\label{intro alignment operator}
    \sum_{I \subseteq [n]} \mu_{I} \rho_{I} \otimes Q^{\otimes I^{c}}.
\end{align}
Each term in the mixture corresponds to a subset \(I \subseteq [n]\) and is a tensor product of a state $\rho_{I}$ of the qudits in $I$ and a fixed state $Q$ on each of the qudits in the complement \(I^c\). Their task was to pick the tuple of states $(\rho_{I})_{I \subseteq [n]}$ so that the overall operator in Eq.~\ref{intro alignment operator} has minimum von Neumann entropy. It is not too difficult to see that an optimal choice has to be one where states in the tuple are pure (see Lem.~\ref{lem: pure states sufficiency} below). However, depending on $\mu$, there may be no such choice where the summands in Eq.~\ref{intro alignment operator} are perfectly aligned. Leditzky et al. conjectured that the von Neumann entropy is minimized by choosing the tuple of states to be $(\ket{q_{1}}\bra{q_{1}}^{\otimes I})_{I \subseteq [n]}$, where $\ket{q_{1}}$ is a maximal eigenvector of $Q$. With this tuple, the summands appear to overlap maximally. They showed that if this conjecture is true, then their channels have additive coherent information. For more information on these channels, see \cite{Leditzky2022a} and the companion paper \cite{Leditzky2022b}.

The structure of this paper is as follows. In Sec.~\ref{sec:prelim}, we introduce notation and the relevant mathematical background. We formally define and motivate the class of spin alignment problems in Sec.~\ref{sec: spin alignment def}. In Sec.~\ref{sec:results}, we provide our results, which include resolutions for various instances of spin alignment. In Sec.~\ref{sec:``dual" problem}, we introduce a general class of optimization problems that is ``dual" to spin alignment. In Sec.~\ref{sec:conclusion}, we conclude and discuss open problems and follow-up lines of inquiry. Appendices App.~\ref{sec:app A} and App.~\ref{sec:app B} contain proofs of technical lemmas. Appendix App.~\ref{sec:app C} contains a resolution of the problem stated in Ex.~\ref{ex: toy example}.

\section{Mathematical preliminaries}
\label{sec:prelim}

In this work, we are concerned with finite-dimensional inner product spaces over $\mathbb{C}$. We denote such spaces by $\mathcal{H},\mathcal{K}$, etc. The dimension of $\mathcal{H}$ is denoted by $d_{\mathcal{H}}$. The set of operators on \(\mathcal{H}\) is denoted by \(\mathcal{B}(\mathcal{H})\). $\mathcal{S}(\mathcal{H}) \subset \mathcal{B}(\mathcal{H})$ denotes the real vector space of self-adjoint operators on $\mathcal{H}$. $\mathcal{P} (\mathcal{H}) \subset \mathcal{S}(\mathcal{H})$ denotes the convex cone of positive semi-definite operators, and $\mathcal{D} (\mathcal{H}) \subset \mathcal{P} (\mathcal{H})$ denotes the convex set of states, that is the set of positive semi-definite operators with unit trace. For $R \in \mathcal{B} (\mathcal{H})$, $\sigma(R) \in \mathbb{R}^{d_{\mathcal{H}}}$ is the vector whose entries are the singular values of $R$ ordered non-increasingly. Similarly, for $S \in \mathcal{S} (\mathcal{H})$, $\lambda(S) \in \mathbb{R}^{d_{\mathcal{H}}}$ is the vector whose entries are the eigenvalues of $S$ ordered non-increasingly. We use the word projector to refer to a self-adjoint projection operator and the word qudit to refer to a physical system that is modeled on a complex inner product space of dimension $d$. For $n \in \mathbb{N}$, the set {\(\{1,\ldots, n\)\}} is denoted by $[n]$ and its power set is denoted by $2^{[n]}$. Given a measure $\mu$ on a countable set, its support is denoted by $\text{supp}(\mu)$. If $x \in \mathbb{R}^{d}$, then $x^{\downarrow}$ and $x^{\uparrow}$ denote the vectors in \(\mathbb{R}^{d}\) with the entries of $x$ ordered non-increasingly and non-decreasingly, respectively. Given $\mathcal{H}$, we distinguish an ordered orthonormal basis $\{ \ket{ q_{i} } \}_{i=1}^{ d_{\mathcal{H}}}$. For $S \in \mathcal{S} (\mathcal{H})$, $S^{\downarrow} := \sum_{i = 1}^{d_{\mathcal{H}}} \lambda_{i} (S) \ket{q_{i}}\bra{q_{i}}$. Real vectors $v_{1}, \ldots, v_{s}$ in $\mathbb{R}^d$ are said to be similarly ordered if there exists a permutation matrix $A: \mathbb{R}^d \rightarrow \mathbb{R}^d$ such that $v_{1}^{\downarrow} = A v_{1}, \ldots, v_{s}^{\downarrow} = A v_{s}$. Similar ordering for real vectors is analogous to perfect alignment for self-adjoint operators. Observe that for an arbitrary vector $u \in \mathbb{R}^d$ and $v \in \text{span} ( (\underbrace{1, \ldots, 1)}_{d \text{ times}})$, the two are similarly ordered. Analogously, an arbitrary self-adjoint operator and an arbitrary operator in the real span of the identity operator are perfectly aligned. This observation shows that neither similar ordering nor perfect alignment is transitive when the dimension of the underlying space is greater than $1$.

Here is a brief review of majorization theory. Given $x, y \in \mathbb{R}^{d}$, $x$ is said to majorize $y$ if
\begin{align} 
\label{eq: majorization cond}
    \sum_{i = 1}^{k} x_{i}^{\downarrow} \geq \sum_{i = 1}^{k} y_{i}^{\downarrow}
\end{align}
for each $k \in [d -1]$ and $\sum_{i = 1}^{d} x_{i}^{\downarrow} = \sum_{i = 1}^{d} y_{i}^{\downarrow}$. This statement is denoted with $x \succeq y$. Majorization is defined on $\mathcal{S}(\mathcal{H})$ spectrally. For $B, C \in \mathcal{S}(\mathcal{H})$, if $\lambda(B) \succeq \lambda(C)$, then $B$ is said to majorize $C$, and we denote that with $B \succeq C$. Majorization induces a preorder on $\mathbb{R}^{d}$ and $\mathcal{S}(\mathcal{H})$, respectively. Given a set, a majorant, if it exists, is an element of the set that majorizes all elements in the set.

There are several useful characterizations of majorization, two of which are relevant to this work. The first is based on doubly-stochastic processing (see Theorem II.1.10 on page 33 in Ref.~\cite{Bhatia1997}): $x$ majorizes $y$ if and only if there exist a tuple of permutation matrices $(P_{i})_{i}$ and a probability measure $(p_{i})_{i}$ such that 
\begin{align} 
\label{eq: cl stoch majorization}
 y = \sum_{i} p_{i} P_{i} x. 
\end{align}
In words, $x$ majorizes $y$ if and only if $y$ belongs to the convex hull of the orbit of $x$ under the action of the symmetric group. That is, $y = A x$ for some doubly-stochastic matrix $A$. Analogously, $B \succeq C$ if and only if
\begin{align}
\label{eq: q stoch majorization}
C = \sum_{i} p_{i} U_{i} B U_{i}^{*}
\end{align}
for a tuple of unitary operators $(U_{i})_{i}$ and a probability measure $(p_{i})_{i}$. Put differently, $C$ is the value of a mixed-unitary quantum channel at $B$.

The second characterization is based on unjust transfers, which are related to the better-known $T$-transforms \cite{Marshall2011}. For distinct $i,j \in [d]$, and $t \in [0,1]$, a $T$-transform (also known as a Robin Hood operation \cite{Arnold2011}) is a linear map with action
 \begin{align}
\label{alg: T-transforms}
(x_{1}, \ldots, x_{i}, \ldots, x_{j}, \ldots, x_{d}) \mapsto (x_{1}, \ldots, (1-t) x_{i} + t x_{j}, \ldots , (1-t) x_{j} + t x_{i}, \ldots, x_{d}).
 \end{align}
 A $T$-transform is a convex combination of the identity and a transposition. The statement $x \succeq y$ holds if and only if $y$ can be obtained from $x$ by a finite number of $T$-transforms (see Lemma B.1 on page 32 in Ref.~\cite{Marshall2011}). A direct consequence of this is that $x \succeq y$ implies $x \oplus z \succeq y \oplus z$ for all $x, y \in \mathbb{R}^{d_{1}}$ and $z \in \mathbb{R}^{d_{2}}$. A transfer on \(\mathbb{R}^d\) is an affine transformation determined by a pair
  of distinct indices \(i,  j \in [d]\) and an amount \(\varepsilon \geq 0\). It has the following action
\begin{align}
\label{alg: transfers}
(x_{1}, \ldots, x_{i}, \ldots, x_{j}, \ldots, x_{d}) \mapsto (x_{1}, \ldots,  x_{i} + \varepsilon, \ldots , x_{j} - \varepsilon, \ldots, x_{d}).
 \end{align}
Here, \(i\) is the receiving entry of the transfer and \(j\) is the giving entry. If \(x_i \geq x_j\), then when applied to \(x\), the transfer is said to be \textit{unjust}.

 \begin{prop}
 \label{prop: T-transforms and unjust transfers are equivalent}
 Let $x, y \in \mathbb{R}^{d}$. Then, $y = T x$ for some $T$-transform $T$ if and only if $x^\downarrow$ can be obtained from $y^\downarrow$ with an unjust transfer. 
\end{prop}

\begin{proof}
Since both $T$-transforms and transfers affect at most a pair of coordinates, we may assume without loss of generality that $d = 2$. Denote $x= (x_{1}, x_{2})$ and $y=(y_{1}, y_{2})$. 

Suppose that there is a \(T\)-transform \(T\) such that \(y=Tx\). Then, for some \(t\in[0,1]\)
\begin{align}
\label{eq: T-transform between pairs}
(y_{1}, y_{2}) = ( (1-t)x_{1} + t x_{2},  (1-t) x_{2} + t x_{1} ).
\end{align}
 Because $y^{\downarrow}_{1}$ is a convex combination of $x_{1}$ and $x_{2}$, we have $x^{\downarrow}_{1} \geq y^{\downarrow}_{1}$. Let $\tilde{\varepsilon} := x^{\downarrow}_{1} - y^{\downarrow}_{1}$, and observe that $x^{\downarrow}_{1} = y^{\downarrow}_{1} + \tilde{\varepsilon}$ and $x^{\downarrow}_{2} = x^{\downarrow}_{1} + x^{\downarrow}_{2} - x^{\downarrow}_{1} = y^{\downarrow}_{1} + y^{\downarrow}_{2} -x^{\downarrow}_{1} = y^{\downarrow}_{2} - \tilde{\varepsilon}$, where we used the fact that $x_{1} + x_{2} = y_{1} + y_{2}$. 

Conversely, suppose there exists an unjust transfer of $\varepsilon > 0$ (the case where $\varepsilon = 0$ is trivial) that takes $y^{\downarrow}$ to $x^{\downarrow}$. That is, $x^{\downarrow}_{1} = y^{\downarrow}_{1} + \varepsilon$ and $x^{\downarrow}_{2} = y^{\downarrow}_{2} - \varepsilon$. Observe that for $\Tilde{t} := \frac{\varepsilon}{(y^{\downarrow}_{1} -y^{\downarrow}_{2}) + 2 \varepsilon}$, 
\begin{align}
\label{eq: T-transform from unjust transfer}    
y^{\downarrow}_{1} = (1-\tilde{t}) x^{\downarrow}_{1} + \tilde{t} x^{\downarrow}_{2}, \quad y^{\downarrow}_{2} = (1-\tilde{t}) x^{\downarrow}_{2} + \tilde{t} x^{\downarrow}_{1}.
\end{align}
Hence, there exists a $T$-transform $T'$ such that $T' x^\downarrow = y^\downarrow $. Without loss of generality, we may take $x = x^\downarrow$. If $y \neq y^\downarrow$, we may compose a transposition with $T'$ to get another $T$-transform $T$ that satisfies $T x = y$. \end{proof}

From this proposition, we conclude that $x \succeq y$ if and only if it is possible to get to $x$ from $y$ with a finite number of transpositions and unjust transfers. In the particular case of $x, y \in \mathbb{R}^{2}$, $x \succeq y$ if and only if $x$ can be obtained from $y$ by an unjust transfer, or an unjust transfer and a transposition. This characterization may be extended to self-adjoint operators as follows. For $B, C \in \mathcal{S} (\mathcal{H})$, $B \succeq C$ if and only if there exists a finite number of transpositions and unjust transfers that take $\lambda(C)$ to $\lambda(B)$.

Let $f$ be a real-valued function such that $\text{dom} (f)$ is a convex set contained in either $\mathbb{R}^{d}$ or $\mathcal{S} (\mathcal{H})$. Suppose that $\text{dom} (f)$ is closed under the action of the symmetric group if it is contained in $\mathbb{R}^{d}$ and closed under the action of the unitary group if it is contained in $\mathcal{S} (\mathcal{H})$. We say that \(f\) is Schur-convex if for all $a, b \in \text{dom} (f)$, $a \succeq b$ implies $f(a) \geq f(b)$. If $-f$ is Schur-convex, then $f$ is said to be Schur-concave. Let $F$ be a real-vector-valued function such that $\text{dom} (F)$ is a convex set contained in either $\mathbb{R}^{d}$ or $\mathcal{S} (\mathcal{H})$. Similarly, suppose that $\text{dom} (F)$ is closed under the action of the symmetric group if it is contained in $\mathbb{R}^{d}$ and closed under the action of the unitary group if it is contained in $\mathcal{S} (\mathcal{H})$. $F$ is called strictly isotone if for all $a, b \in \text{dom} (F)$ such that $a \succeq b$, it holds that $F(a) \succeq F(b)$.

An important tool in our analysis is Fan's maximum principle \cite{Fan1949} for self-adjoint operators. It states that for $T \in \mathcal{S}(\mathbb{C}^{d})$, $k \in [d]$,
\begin{align} 
\label{eq: Fan maximum prinicple}
    \sum_{i = 1}^{k} \lambda_{i} (T) = \max\{\tr(TP) \; | \;  \text{\(P\) is a rank-\(k\) projector}\}.
\end{align}
It follows directly from this principle that
\begin{align} 
\label{eq: Fan majorization rel self-adjoint}
    \sum_{i=1}^{s} S_{i} \preceq \sum_{i=1}^{s} S_{i}^{\downarrow}
\end{align}
for any self-adjoint $S_{1},\ldots, S_{s}$. Moreover, the inequalities above are simultaneously satisfied with equality if and only if \(S_{1},\ldots, S_{s}\) are perfectly aligned.

We are interested in maximizing continuous functions over $\mathcal{D}(\mathcal{H}^{\otimes n})$, specifically those that are convex and unitarily invariant. Of particular interest are unitarily invariant norms. Generically denoted by $||| \cdot |||$, these are operator norms that satisfy $||| \cdot  ||| = ||| U (\cdot)  V ||| $ for all unitaries $U$ and $V$. For an operator $R \in \mathcal{B}(\mathbb{C}^{d})$, its Schatten $p$-norm for $p \in [1, \infty)$ is defined as 
\begin{align}
\label{eq: p-Schatten norm}
    || R ||_{p} := \tr (|R|^{p})^{\frac{1}{p}}, 
\end{align}
where $|R|$ denotes the absolute value $\sqrt{R^{*} R}$. Schatten norms satisfy $\lim_{p \rightarrow \infty} || \cdot ||_{p} = || \cdot ||$, where $|| \cdot ||$ is the operator norm. Both Rényi entropies and Tsallis entropies \cite{Renyi1960, Tsallis1988} are monotonic functions of Schatten $p$-norms for $p \in (1, \infty)$. For $R \in \mathcal{B}(\mathbb{C}^{d})$, $k \in [d]$, the Fan norm of order $k$ is the sum of the $k$ largest singular values of $R$:
\begin{align}
\label{eq: k-Fan norm} 
    || R ||_{(k)} := \sum_{i = 1}^{k} \sigma_{i} (R).
\end{align}
For positive semi-definite $R$, $|| R ||_{(k)}$ is the sum of its $k$ largest eigenvalues. Both Schatten norms and Fan norms are unitarily invariant norms. For any $A,B \in \mathcal{B}(\mathbb{C}^{d})$, $|| A ||_{(k)} \leq || B ||_{(k)}$ for all $k \in [d]$ implies $||| A ||| \leq ||| B |||$ for every unitarily invariant norm $||| \cdot |||$. This is Fan's dominance theorem (see Theorem IV.2.2 on page 93 in Ref.~\cite{Bhatia1997}).

For readability, we index factors in tensor products with subsets,
each denoting the domain of the corresponding factor.  For
  example, given \(O \in \mathcal{B} ( \mathcal{H})\) and
  \(J \subseteq [n]\), \(O^{\otimes J}\) denotes the fully
  factorizable operator acting on \(|J|\) specific tensor copies of
  \(\mathcal{H}\)---specified to be corresponding to \(J\)---with each
  of its factors is equal to \(O\).  In general, we do not write the
curly brackets of sets in subscripts.
 
\section{Spin alignment problems and the alignment ordering}
\label{sec: spin alignment def}

\begin{definition} (Spin alignment problem).
    Let a local dimension $d$, a number of qudits $n$, a probability measure $\mu$, and a state $Q = \sum_{i=1}^{d} \lambda_{i} (Q) \ket{q_{i}}\bra{q_{i}} \in \mathcal{D}(\mathbb{C}^{d})$ be given. If $f$ is a continuous, unitarily invariant, convex function whose domain includes $\mathcal{D}((\mathbb{C}^{d})^{\otimes n})$, then the spin alignment problem associated with $(d, n, \mu, Q, f)$ is to maximize the objective function $f( \sum_{I \subseteq [n]} \mu_{I} \rho_{I} \otimes Q^{\otimes I^{c}})$ over all state tuples $(\rho_{I})_{I \subseteq [n]}$. The operator argument
    \begin{align} 
\label{alignment operator}
    \sum_{I \subseteq [n]} \mu_{I} \rho_{I} \otimes Q^{\otimes I^{c}}
\end{align}
is called an \textit{alignment operator}. 
\end{definition}

Since the objective function is continuous and the set of alignment operators is compact, an optimal point exists. For concave functions, such as the von Neumann entropy, the associated spin alignment problem is a minimization problem. The orthonormal basis $(\ket{q_{i}})_{i = 1}^{d}$ and bases arising from taking tensor powers of it are called \textit{computational} bases. Operators that are diagonal in these bases are called \textit{classical}. We refer to assertions that the state tuple $(\ket{q_{1}}\bra{q_{1}}^{\otimes I})_{I \subseteq [n]}$ is optimal as spin alignment conjectures. We mostly focus on the dependence of the problems on the state $Q$ and the objective function $f$ and implicitly take the remaining parameters to be arbitrary. For example, when we speak of the spin alignment problem associated with a state $Q$, we mean the collection of all spin alignment problems with state $Q$.

We conjecture that the spectrum of an alignment operator with state tuple $(\ket{q_{1}}\bra{q_{1}}^{\otimes I})_{I \subseteq [n]}$ majorizes the spectra of alignment operators arising with other state tuples for the same $Q$ and $\mu$. That is, it has the largest possible value of the order-$k$ Fan norm for each $k \in [d^{n}]$. By the doubly-stochastic processing characterization of majorization, if this conjecture is true, then every spin alignment conjecture is true. We call this the \textit{strong} spin alignment conjecture. In cases where $Q$ is pure, the conjecture follows directly from Eq.~\ref{eq: Fan majorization rel self-adjoint}. So, we focus on the cases where $Q$ is mixed.

Spin alignment problems can be extended to cases where the operators \(Q\) are not the same in each term in the sum in Eq.~\ref{alignment operator}. For example, spin alignment conjectures can be considered in the more general case where the alignment operators are of the form 
\begin{align}
\label{eq: general alignment operator}
\sum_{I \subseteq [n]} \mu_{I} \rho_{I} \otimes Q_{I^{c}},
\end{align}
where for $I^{c} \subseteq [n], Q_{I^{c}} = \bigotimes_{j \in I^{c}} Q_{I^{c},j}$ such that for each $ j \in [n]$, the operators in the tuple $(Q_{I^{c}, j})_{I^{c} \subseteq [n]}$ are perfectly aligned.
In this work, we consider only alignment operators as in Eq.~\ref{alignment operator}, though many of our results apply for these more general alignment operators.

\begin{lem}
\label{lem: pure states sufficiency}
For any spin alignment problem, it suffices to restrict to state tuples of pure states. 
\end{lem}
\begin{proof}
Notice that the alignment operator is a linear function of the state tuple, so convex combinations of state tuples correspond to convex combinations of alignment operators. By convexity of the objective function, there exists an extremal point that is optimal. The extremal points of state tuples are pure-state tuples. 
\end{proof}

\begin{remark}
    If the spin alignment conjecture for a function $g$ is true, then it also holds that $f = h \circ g$ attains its maximum over alignment operators at $(\ket{q_{1}}\bra{q_{1}}^{\otimes I})_{I \subseteq [n]}$ whenever $h$ is monotonically increasing. If $h$ is monotonically decreasing, then $f$ attains a minimum there. Examples of this include Schatten norms and the corresponding 
     Rényi entropies.
\end{remark}

In addition to showing additivity of coherent information for platypus channels, proving the strong spin alignment conjecture would formalize intuitions we have about how dispersion behaves under mixing. Namely, to minimize dispersion in a mixture, the optimal procedure ought to be one where each signal has the lowest dispersion possible and the signals collectively overlap as much as possible. This can be considered a refinement of Fan's majorization relation Eq.~\ref{Fan's majorization relation} because in the spin alignment problem, while each summand has minimum dispersion by a pure state choice, we are not always free to choose the bases so that they are perfectly aligned.

When operators cannot be chosen to be perfectly aligned to achieve minimum dispersion, we need a way of comparing the alignment of different tuples of operators. The next definition provides such a way while restricting the tuples to have matching spectra.
\begin{definition}
\label{def: alignment} (Alignment)
    Let $\mathcal{T} = (T_{1}, T_{2}, \ldots, T_{\ell})$ and $\mathcal{R} = (R_{1}, R_{2}, \ldots, R_{\ell})$ be two tuples of self-adjoint operators such that for each $i \in [\ell], \lambda(T_{i}) = \lambda(R_{i})$. Then, $\mathcal{T}$ is \textit{more aligned} than $\mathcal{R}$ if for all probability measures $(p_{i})_{i =1}^{\ell}$, the following majorization relation holds
    \begin{align}
    \label{def: partial alignment}
    \sum_{i =1}^{\ell} p_{i} T_{i}  \succeq   \sum_{i =1}^{\ell} p_{i} R_{i}. 
    \end{align}
\end{definition}
Alignment can be used to construct a preorder on any set of $\ell$-tuples of operators with matching spectra. We are interested in cases where there exists a maximal element according to this ordering. The basic example is one where the set in question contains a tuple of perfectly aligned operators. The strong spin alignment conjecture may be phrased as follows. 
\begin{conj}
\label{conj: strong spin alignment conejcture}
    The tuple $(\ket{q_{1}}\bra{q_{1}}^{\otimes I} \otimes Q^{\otimes I^{c}})_{I \subseteq [n]}$ is more aligned than every state tuple of the form $(\ket{\psi_{I}}\bra{\psi_{I}}_{I} \otimes Q^{\otimes I^{c}})_{I \subseteq [n]}$.
\end{conj}

A key reason to consider alignment as opposed to perfect alignment is that the latter does not behave well with respect to tensor multiplication. Specifically, the coordinate-wise tensor product of two perfectly aligned tuples is not necessarily perfectly aligned. Even worse, the perfect alignment of $A$ and $B$ does not imply the perfect alignment of $C \otimes A$ and $C \otimes B$ for self-adjoint $A, B$, and $C$. 
\begin{example}
\label{ex: tensor product and perfect alignment}
Consider $D = \ket{q_1}\bra{q_1} + \frac{1}{2} \ket{q_2}\bra{q_2} + \frac{1}{3} \ket{q_3}\bra{q_3}$. It is clear that $D$ and $\ket{q_1}\bra{q_1}$ are perfectly aligned. However, $D \otimes D$ and $D \otimes \ket{q_1}\bra{q_1}$ are not perfectly aligned. 
\end{example}
In contrast, $(C \otimes A^{\downarrow}, C \otimes B^{\downarrow})$ is more aligned than $(C \otimes A, C \otimes B)$. This is a simple consequence of Eq.~\ref{eq: Fan majorization rel self-adjoint} and the following lemma, which appears in \cite{Bondar2003}. We include a proof of the lemma for completeness. 

\begin{lem}
\label{lem: tensor product and majorization}
Let $A, B \in \mathcal{S}(\mathcal{H})$ and $C \in \mathcal{S} (\mathcal{K})$. If $A \succeq B$, then $C \otimes A \succeq C \otimes B$.
\end{lem}
\begin{proof}
This follows from the doubly stochastic characterization of majorization. $A \succeq B$ is equivalent to the existence of a mixed-unitary quantum channel $\mathcal{N}$ such that $B = \mathcal{N} (A)$. It is clear that $\id \otimes \mathcal{N} 
$ is mixed-unitary and $\id \otimes \mathcal{N} (C \otimes A) = C \otimes B$.\end{proof}
We note that the converse of this lemma does not hold. Aside from trivial counterexamples where $C = 0$, others arise in the context of catalyst-enabled quantum pure-state interconversion with LOCC~\cite{Daniel1999}.

If we restrict the factors on the left to be positive semi-definite, we can prove the following statement.
\begin{prop}
\label{prop: alignment with commuting factor}
    Let $C_{1}, \ldots , C_{\ell} \in  \mathcal{P}(\mathcal{K})$ be simultaneously diagonalizable. For every $A_{1}, \ldots, A_{\ell} \in \mathcal{S}(\mathcal{H})$, $(C_{i} \otimes A^{\downarrow}_{i})_{i=1}^{\ell}$ is more aligned than $(C_{i} \otimes A_{i})_{i=1}^{\ell}$.
\end{prop}
\begin{proof}
Let $p = (p_{i})_{i=1}^{\ell}$ be an arbitrary probability measure.  Let $\{ \ket{j} \}_{j = 1}^{d_{\mathcal{K}}}$ be an orthonormal basis for $\mathcal{K}$ where $C_{1}, \ldots. C_{\ell}$ are simultaneously diagonal. For each $i \in [\ell]$, let $\sum_{j=1}^{d_{\mathcal{K}}} c_{j | i} \ket{j}\bra{j}$ be the corresponding spectral decomposition of $C_{i}$. Observe that
\begin{align}
\label{eq: conditional mixed-unitary}
\sum_{i = 1}^{\ell} p_{i} C_{i} \otimes A^{\downarrow}_{i} = \sum_{i = 1}^{\ell} p_{i} \sum_{j=1}^{d_{\mathcal{K}}} c_{j|i} \ket{j}\bra{j} \otimes A^{\downarrow}_{i} = \sum_{j=1}^{d_{\mathcal{K}}} \ket{j}\bra{j} \otimes (\sum_{i = 1}^{\ell} p_{i} c_{ j| i} A^{\downarrow}_{i}).
\end{align}
For each $j \in [d_{\mathcal{K}}]$, we know from Eq.~\ref{eq: Fan majorization rel self-adjoint} that $\sum_{i = 1}^{\ell} p_{i} c_{j|i} A^{\downarrow}_{i} \succeq \sum_{i = 1}^{\ell} p_{i} c_{j|i} A_{i}$. Hence, there exists a mixed-unitary channel $\mathcal{N}_{j}$ that takes the left-hand side to the right-hand side. Let \(\mathcal{N}\) be the global channel that conditional on \(j\) applies \(\mathcal{N}_{j}\) for each $j \in [d_{\mathcal{K}}]$. It is a mixed-unitary quantum channel and it takes $\sum_{i = 1}^{\ell} p_{i} C_{i} \otimes A^{\downarrow}_{i}$ to $\sum_{i = 1}^{\ell} p_{i} C_{i} \otimes A_{i}$. \end{proof}

The strong spin alignment conjecture subsumes a conjecture about maximizing unitarily invariant norms at the output of a tensor product of depolarizing channels. For $q \in [\frac{-1}{d^{2} -1}, 1]$, the qudit depolarizing channel $\Delta_{d, q}$ has the action: 
\begin{align}
\label{eq: depolarizing channel def}  
\Delta_{d, q} ( \cdot ) := q\, \id (\cdot) + (1- q) \frac{\mathds{1}}{d} \tr( \cdot ).
\end{align}
In \cite{King2003}, to show additivity of the Holevo information of depolarizing channels, King proves that for any $p \in [1,\infty)$, the Schatten $p$-norm of the output of a tensor product of depolarizing channels is maximized at a pure product input. It is reasonable to extrapolate and suspect that this holds for all unitarily invariant norms. The basic idea is that any entanglement in the input manifests as excess dispersion at the output. For $i \in [n]$, let $q_{i} \in [0,1]$ and notice that 
\begin{align}
\label{eq: depolarizing to alignment} 
\bigotimes_{i=1}^{n} \Delta_{d, q_{i}} (\rho_{[n]}) = \sum_{I \subseteq [n]} \nu_{I} \rho_{I} \otimes (\frac{\mathds{1}}{d})^{\otimes I^{c}},
\end{align}
where \( \nu_{I}= \prod_{i\in I} q_{i} \prod_{j\in I^{c}} (1-q_{j})\) and $\rho_{I}$ denotes the marginal state on $I$ of $\rho_{[n]}$. This is an alignment operator with $Q = \frac{\mathds{1}}{d}$ arising from a state tuple satisfying an instance of the quantum marginal problem \cite{Schilling2014}. Since the conjectured optimal point $(\ket{q_{1}}\bra{q_{1}}^{\otimes I})_{I \subseteq [n]}$ satisfies an instance of the quantum marginal problem, it follows that a proof of the strong spin alignment conjecture for $Q = \frac{\mathds{1}}{d}$ implies a resolution for this problem. This problem is analogous to the now-solved problem of maximizing unitarily invariant output norms for single-mode Gaussian channels. Specifically, it was shown in \cite{Mari2014} that coherent states are optimal.

\section{Results}

\label{sec:results}
\subsection*{A. First observations}

First, we show that the spin alignment conjecture is true for the case where the objective function is $\lambda_{1}$. 
\begin{prop} \label{maximum eigenvalue case}
     $\lambda_{1} ( \sum_{I \subseteq [n]} \mu_{I} \rho_{I} \otimes Q^{\otimes I^{c}})$ is maximized at $(\rho_{I})_{I \subseteq [n]} = (\ket{q_{1}}\bra{q_{1}}^{\otimes I})_{I \subseteq [n]}$.
\end{prop}
\begin{proof}
    This follows from Eq.~\ref{eq: Fan maximum prinicple}. Specifically, 
    \begin{align}
    \label{eq: inequalities for spectral radius}
        \lambda_{1} ( \sum_{I \subseteq [n]} \mu_{I} \rho_{I} \otimes Q^{\otimes I^{c}}) \leq  \sum_{I \subseteq [n]} \mu_{I} \lambda_{1}(\rho_{I} \otimes Q^{\otimes I^{c}})
        \leq  \sum_{I \subseteq [n]} \mu_{I} \lambda_{1}(Q^{\otimes I^{c}}).
    \end{align}
    At $(\ket{q_{1}}\bra{q_{1}}^{\otimes I})_{I \subseteq [n]}$, both inequalities are satisfied with equality. \end{proof}
This proves one of $d^{{n}} - 1$ conditions in Eq.~\ref{eq: majorization cond} necessary to prove the strong spin alignment conjecture. The difficulty in proving the remaining ones lies in the fact that summands of alignment operators with pure-state tuples cannot be made perfectly aligned in general.

As explained in the previous section, the strong spin alignment conjecture holds for pure \(Q\). Next, for $Q$ of rank $2$ or more, we prove a reduction to cases where $\lambda_{1} (Q) = \lambda_{2} (Q)$.
\begin{prop} \label{reduction to flatter spectrum}
Suppose $Q = \sum_{i=1}^{d} \lambda_{i} (Q) \ket{q_{i}}\bra{q_{i}} \in \mathcal{D} (\mathcal{H})$ has rank $\geq 2$. Define
\begin{align}
\label{def: tilde Q for mixed Q}
    \tilde{Q} = \frac{1}{1 - (\lambda_{1} (Q) - \lambda_{2} (Q))} \Big(\lambda_{2} (Q) (\ket{q_{1}}\bra{q_{1}} + \ket{q_{2}}\bra{q_{2}}) + \sum_{i=3}^{d} \lambda_{i} (Q) \ket{q_{i}}\bra{q_{i}}\Big). 
\end{align}
If for all $\tilde{\mu}$, $f( \sum_{I \subseteq [n]} \tilde{\mu}_{I} \rho_{I} \otimes \tilde{Q}^{\otimes I^{c}})$ is maximized at $(\rho_{I})_{I \subseteq [n]} = (\ket{q_{1}}\bra{q_{1}}^{\otimes I})_{I \subseteq [n]}$, then for all $\mu$, $f( \sum_{I \subseteq [n]} \mu_{I} \rho_{I} \otimes Q^{\otimes I^{c}})$ is maximized there as well. 
\end{prop} 
\begin{proof} 
Denote $\varepsilon := \lambda_{1} (Q) - \lambda_{2} (Q)$. Observe that 
$Q = (1- \varepsilon) \tilde{Q} + \varepsilon \ket{q_{1}}\bra{q_{1}}$. That is, $Q$ lies on the line segment connecting $\tilde{Q}$ and $\ket{q_{1}}\bra{q_{1}}$. For any probability measure $\mu$, we can express
\begin{align}
\label{eq: expanded Q into tilde Q and q_1}
    \sum_{I \subseteq [n]} \mu_{I} \rho_{I} \otimes Q^{\otimes I^{c}} = \sum_{K \subseteq [n]} \tilde{\mu}_{K} \tilde{\rho}_{K} \otimes \tilde{Q}^{\otimes K^{c}},
\end{align}
where each factor of \(Q\) on the left-hand side was replaced
  by \((1- \varepsilon) \tilde{Q} + \varepsilon \ket{q_{1}}\bra{q_{1}}\) and \( \mu_{I} \rho_{I} \otimes \varepsilon^{|J|} (1-\varepsilon)^{|(I \cup J)^{c}|} \ketbra{q_{1}}^{\otimes J}\) for \(I\cup J=K\) were summed to obtain \(\tilde{\mu}_{K}\tilde{\rho}_{K}\). By assumption,
\begin{align}
\label{estimates for Q and Q tilde alignment problems}
f(\sum_{I \subseteq [n]} \mu_{I} \rho_{I} \otimes Q^{\otimes I^{c}}) &= f(\sum_{K \subseteq [n]} \tilde{\mu}_{K} \tilde{\rho}_{K} \otimes \tilde{Q}^{\otimes K^{c}})\\
&\leq f(\sum_{K \subseteq [n]} \tilde{\mu}_{K} \ket{q_{1}}\bra{q_{1}}^{\otimes K} \otimes \tilde{Q}^{\otimes K^{c}})  = f(\sum_{I \subseteq [n]} \mu_{I} \ket{q_{1}}\bra{q_{1}}^{\otimes I} \otimes Q^{\otimes I^{c}}), 
\end{align} 
where we reversed the process used to obtain the right-hand side
of Eq.~\ref{eq: expanded Q into tilde Q and q_1} for \((\rho_{I})_{I \subseteq [n]} =(\ketbra{q_{1}}^{\otimes I})_{I \subseteq [n]}\) to arrive at the last equality. \end{proof}
\begin{remark}
For $d = 2$, this reduction implies that $Q = \frac{\mathds{1}}{2}$ may be assumed without loss of generality. One can therefore apply the argument in Sec.~VI.B of~\cite{Leditzky2022a} to show that the strong spin alignment conjecture holds for \(n=d=2\). 
\end{remark}

The last result of this subsection relates to problem instances where the union of the elements of the support of the probability measure $\mu$ is not $[n]$.  That is, those instances where $\mu$ is such that there exists nonempty $K^{c} \subseteq [n]$ such that $\sum_{I \subseteq [n]} \mu_{I} \rho_{I} \otimes Q^{\otimes I^{c}} =  (\sum_{J \subseteq K} \mu_{J} \rho_{J} \otimes Q^{\otimes J^{c}}) \otimes Q^{\otimes K^{c}}$. If the objective function is multiplicative under tensor multiplication, such as Schatten norms \cite{AUBRUN_2011}, then we may without loss of generality trace over the qudits in $K^{c}$ and consider the optimization problem on the rest. Not all Fan norms are multiplicative under tensor multiplication. However, to prove the strong spin alignment conjecture, it may be assumed without loss of generality that the union of the support of \(\mu\) is \([n]\). 
\begin{prop}
\label{prop: union of mu support is maximal}   
   For every $n_{1}, n_{2} \in \mathbb{N}$, if the strong spin alignment conjecture holds for $\sum_{I \subseteq [n_{1}]} \mu_{I} \rho_{I} \otimes Q^{\otimes I^{c}}$, then it holds for $\sum_{I \subseteq [n_{1}]} \mu_{I} \rho_{I} \otimes Q^{\otimes I^{c}} \otimes Q^{\otimes n_2}$.
\end{prop}
\begin{proof} This a direct consequence of Lem.~\ref{lem: tensor product and majorization}. \end{proof}

\subsection*{B. Classical spin alignment}
\label{subsec: classical spin alignment}
We prove a classical version of the strong spin alignment conjecture. That is, we prove it with the extra assumption that $(\rho_{I})_{I \subseteq [n]}$ is such that $\rho_{I}$ is diagonal in the computational basis for each $I \subseteq [n]$.

\begin{prop} \label{diagonal spin alignment}
    For every probability measure $\mu$ and  tuple of classical states $(\rho_{I})_{I \subseteq [n]}$,
    \begin{align}
    \label{rel: classical spin alignment}
             \sum_{I \subseteq [n]} \mu_{I} \rho_{I} \otimes Q^{\otimes I^{c}} \preceq   \sum_{I \subseteq [n]} \mu_{I} \ket{q_{1}}\bra{q_{1}}^{\otimes I} \otimes Q^{\otimes I^{c}}, 
    \end{align}  
\end{prop}
\begin{proof}
    By the same argument as in Lem.~\ref{lem: pure states sufficiency}, without loss of generality, the state tuple $(\rho_{I})_{I \subseteq [n]}$ may be taken to consist of pure classical states. For $I \subseteq [n]$, we denote $\rho_{I} = \ket{t_{I}}\bra{t_{I}}$, where $t_{I}$ is a string of length $|I|$ over the alphabet $[d]$. For example, $ \ket{1 2 3}\bra{1 2 3}  := \ket{q_{1}}\bra{q_{1}} \otimes  \ket{q_{2}}\bra{q_{2}} \otimes  \ket{q_{3}}\bra{q_{3}}$. 
    
    The idea of the proof is to start with an arbitrary classical assignment $(t_{I})_{I \subseteq [n]}$ and flip to 1 sequentially to improve the assignment making use of the transitivity of the majorization ordering. We start by considering the first qudit. We write the alignment operator in the form 
    \begin{align}
    \label{eq: decompostion of alignment operator}
        \sum_{I \subseteq [n]} \mu_{I} \ket{t_{I}}\bra{t_{I}} \otimes Q^{\otimes I^{c}} = Q \otimes B + \ket{1}\bra{1} \otimes C + \sum_{j = 2}^{d} \ket{j}\bra{j} \otimes D_{j}, 
    \end{align}
    where \(B\), \(C\), and \(D_{j}\) can be expressed as
    \begin{align}
        B &= \sum_{I \subseteq [n] \setminus \{1\} } \nu_{I} \ket{t'_{I}} \bra{t'_{I}} \otimes Q^{I^{c}},\\
        C &=  \sum_{I \subseteq [n] \setminus \{1\} } w_{I} \ket{t''_{I}} \bra{t''_{I}} \otimes Q^{I^{c}}, \\
        D_{j} &=  \sum_{I \subseteq [n] \setminus \{1\} } \kappa_{j,I} \ket{t_{j,I}} \bra{t_{j,I}} \otimes Q^{I^{c}}, \label{lem:overlapDj}
    \end{align}
    where the summand $Q\otimes B$ accounts for all $t_{I}$ with
      $1\not\in I$, $\ketbra{1}\otimes C$ accounts for all $t_{I}$
      with $1\in I$ and first letter $1$ (the ``good'' $t_{I}$)
      and the $D_{j}$ account for the rest. The measures $\nu, w,$ and $\kappa_{j}$
    are positive, and the operators $B, C,$ and $D_{j}$ are positive
    semi-definite and classical. We consider each
    \(j \in [d] \setminus \{1\}\) in sequence and flip the first
    entries of those \(t_{I}\) contributing to $D_{j}$ from \(j\) to \(1\). This implements the following change
      in the classical assignment $(t_{I})_{I\in\subseteq[n]}$: For
      every $t_{I}$ of the form
      $ t_{I}= \ket{j}\bra{j} \otimes \ket{t_{j,I}} \bra{t_{j,I}}$
      contributing to the summand $\ket{j}\bra{j}\otimes D_{j}$,
      replace it with
      $\ket{1}\bra{1} \otimes \ket{t_{j,I}} \bra{t_{j,I}}$.  The
      representation in Eq.~\ref{eq: decompostion of alignment
        operator} changes according to the assignments
      $D_{1}\gets D_{1}+D_{j}$ and $D_{j}\gets 0$, so that the new
      alignment operator is compatible with $Q$ and $\mu$.
    For a string $r$ of length $n-1$ over $[d]$, let $b_{r}, c_{r},$ and $d_{j, r}$ be the eigenvalues for the eigenstate $\ket{r}$ of $B, C,$ and $D_{j}$ before the change. The change in the spectrum can be implemented by composing transfers of the form: $d_{j, r}$ is taken from the eigenvalue of $\ket{j}\otimes\ket{r}$ and given to the eigenvalue of $\ket{1} \otimes \ket{r}$ for each string $r$. These transfers may be implemented in parallel. For a string $r$, those two entries in the spectrum before the change are
      \begin{align}
      \label{eq: old entires}
      (q_{1} b_{r} + c_{r},  q_{j} b_{r} + d_{j,r}),  
      \end{align}
      while after, they are
      \begin{align}
      \label{eq: new entires}
      (q_{1} b_{r} + c_{r} + d_{j, r}, q_{j} b_{r}). 
      \end{align}
      To see that the spectrum after the transfer majorizes the one before it, observe that
    \begin{align}
        q_{1} b_{r} + c_{r} + d_{j,r} \geq \max{(q_{1} b_{r} + c_{r},q_{j} b_{r} + d_{j,r})}. 
    \end{align}
    That is, the larger value of the pair did not decrease after the transfer. Hence, the transfer is either unjust or its composition with a transposition is unjust. By transitivity of majorization, the same process may now be done to the remaining qudits. \end{proof}

\subsection*{C. Schatten norms of integer order}

We prove the spin alignment conjecture for Schatten norms of integer order. By extension, this implies that any monotonic function of these norms, such as the corresponding Rényi entropies, is optimized at the conjectured optimal state tuple $(\ket{q_{1}}\bra{q_{1}}^{\otimes I})_{I \subseteq [n]}$. The proof relies on the fact that overlaps between summands in an alignment operator are simultaneously maximized at the conjectured optimal state tuple. Here, overlap means the absolute value of the trace of any product of these summands. In fact, the overlap lemma Lem.~\ref{overlap lemma} in App.~\ref{sec:app A} establishes the stronger statement that every unitarily invariant norm of such a product is maximized at $(\ket{q_{1}}\bra{q_{1}}^{\otimes I})_{I \subseteq [n]}$.  

\begin{cor} 
\label{cor: ovelap of absolute value of trace}
     Let $I_{1},\ldots,I_{\ell}$ be a family of subsets of $[n]$. For any tuple $(R_{I_{i}})_{i=1}^{\ell}$ of operators each with trace norm at most 1, the following inequality holds
    \begin{align} 
    \label{ineq: absolute value of trace}
        \left| \tr( \;\prod_{i=1}^{\ell}  R_{I_{i}} \otimes Q^{\otimes I_{i}^{c}}) \right| \leq   \tr( \;\prod_{i=1}^{\ell} \ket{q_{1}}\bra{q_{1}}^{\otimes I_{i}} \otimes Q^{\otimes I_{i}^{c}}). 
    \end{align}
\end{cor}
\begin{proof}
The result follows from an application of Lem.~\ref{overlap lemma} in the following chain of inequalities:
  \begin{align}
\label{ineq: chain of estimates for trace of absolute value.}
    \left| \tr( \;\prod_{i=1}^{\ell} R_{I_{i}} \otimes Q^{\otimes I_{i}^{c}})\right| &\leq \tr( \left| \;\prod_{i=1}^{\ell} R_{I_{i}} \otimes Q^{\otimes I_{i}^{c}}\right|) \nonumber\\
    &\leq \tr( \left|\;\prod_{i=1}^{\ell} \ket{q_{1}}\bra{q_{1}}^{\otimes I_{i}} \otimes Q^{\otimes I_{i}^{c}}\right|) \hspace{.25in} \textrm{by Lem.~\ref{overlap lemma}}\nonumber\\
    &=  \tr( \;\prod_{i=1}^{\ell} \ket{q_{1}}\bra{q_{1}}^{\otimes I_{i}} \otimes Q^{\otimes I_{i}^{c}}). 
\end{align} \end{proof}

\begin{thm}
\label{thm: spin alignment for Schatten norms of integer order}
For arbitrary $m \in \mathbb{N}$, the following inequality holds
\begin{align}
    || \sum_{I \subseteq [n]} \mu_{I} \rho_{I} \otimes Q^{\otimes I^{c}} ||_{m} \leq || \sum_{I \subseteq [n]} \mu_{I} \ket{q_{1}}\bra{q_{1}}^{\otimes I} \otimes Q^{\otimes I^{c}} ||_{m}
\end{align}
\end{thm}
 \begin{proof}
   Notice that $ || \sum_{I \subseteq [n]} \mu_{I} \rho_{I} \otimes Q^{\otimes I^{c}} ||_{m}^{m}$ is a linear combination of traces of
   products of operators of the form \(\rho_{I} \otimes Q^{\otimes I^{c}}\) for some $I \subseteq [n]$. By Corollary \ref{cor: ovelap of absolute value of trace}, the absolute value of each term is maximized at
   $(\ket{q_{1}}\bra{q_{1}}^{\otimes I})_{I \subseteq [n]}$. By the triangle inequality for the absolute value and the fact that the terms in the combination have equal phase at $(\ket{q_{1}}\bra{q_{1}}^{\otimes I})_{I \subseteq [n]}$, 
   \begin{align}
        || \sum_{I \subseteq [n]} \mu_{I} \rho_{I} \otimes Q^{\otimes I^{c}} ||_{m}^{m} \leq || \sum_{I \subseteq [n]} \mu_{I} \ket{q_{1}}\bra{q_{1}}^{\otimes I} \otimes Q^{\otimes I^{c}} ||_{m}^{m}.
   \end{align}
   The claim follows from the fact that $m$th root function is monotonically non-decreasing on $[0,\infty)$. \end{proof}

We note here that generalizing this result to Schatten norms of arbitrary order would imply the original spin alignment conjecture as stated in \cite{Leditzky2022a}. However, even that would not be sufficient for proving Conj.~\ref{conj: strong spin alignment conejcture} as can be seen from the equivalence of catalytic majorization and an infinite set of inequalities for (generalized) Rényi entropies \cite{klimesh2007}.

\subsection*{D. The case where $ \text{supp} (\mu) \leq 2$ and $Q$ is proportional to a projector}
\label{subsec: two projectors}

We now prove the strong spin alignment conjecture for alignment operators of the form
\begin{align} 
\label{spin alignment with mu 2}
(1 - t) \ket{\psi_{{1}}}\bra{\psi_{{1}}}_{I_{1}} \otimes Q^{\otimes I_{1}^{c}} + t \ket{\psi_{{2}}}\bra{\psi_{{2}}}_{I_{2}}  \otimes Q^{\otimes I_{2}^{c}}, 
\end{align}
where $t \in [0,1]$, $I_{1}, I_{2} \subseteq [n]$, and $Q$ is proportional to a projector. Since we have established that we may take $I_{1} \cup I_{2} = [n]$ without loss of generality (see Lem.~\ref{prop: union of mu support is maximal}), the cases where $I_{1} \subseteq I_{2}$ or $I_{1} \supseteq I_{2}$ are simply resolved by Eq.~\ref{eq: Fan majorization rel self-adjoint}. To address the other cases, we consider the more general problem of the alignment of two projectors subject to a certain overlap constraint.

Let $\mathcal{K}$ be given. For each tuple $(r_{1}, r_{2}, c)$, where $r_{1}, r_{2} \in [d_{\mathcal{K}}] \cup \{0\}$ and $c \geq 0$, let $G_{r_{1}, r_{2}, c}$ denote the set of pairs of projectors $(P_1, P_2)$ satisfying the conditions:  $\text{rank}(P_{1}) = r_{1}$, $\text{rank}(P_{2}) = r_{2}$, and $\tr(|P_{1} P_{2}|) \leq c$.

\begin{lem}
\label{lem: non-empty criterion for feasible sets.}
$G_{r_{1}, r_{2}, c}$ is nonempty if and only if 
$r_{1} + r_{2} - d_{\mathcal{K}} \leq \lfloor c \rfloor$.
\end{lem}
\begin{proof}
Suppose $r_{1} + r_{2} - d_{\mathcal{K}} > \lfloor c \rfloor$. If $P_{1}$ and $P_{2}$ are projectors satisfying $\text{rank} (P_{1}) = r_{1}$ and $\text{rank} (P_{2}) = r_{2}$, then from basic linear algebra 
\begin{align}
\label{ineq: basic linear algebra dim rel}
    \dim( \text{supp}(P_{1}) \cap \text{supp}(P_{2})) \geq r_{1} + r_{2} - d_{\mathcal{K}} > \lfloor c \rfloor.
\end{align}
Hence, $P_{1} P_{2}$ acts as the identity on a subspace of dimension at least $\lfloor c \rfloor + 1$ and so $\tr(|P_{1} P_{2}|) > c$. This implies $G_{r_{1}, r_{2}, c}$ is empty.

If on the other hand $r_{1} + r_{2} - d_{\mathcal{K}} \leq \lfloor c \rfloor$, let $P$ be a projector of rank equal to $\max(r_{1} + r_{2} - d_{\mathcal{K}}, 0)$. Consider the two projectors $P_{1} = P + R_{1}$ and $P_{2} = P + R_{2}$, where $R_{1}, R_{2}$ are projectors satisfying $R_{1} \perp P, R_{2} \perp P, R_{1} \perp R_{2}$, and $\text{rank}(R_{1}) = r_{1} - \text{rank}(P), \text{rank}(R_{2}) = r_{2} - \text{rank}(P)$. Then $P_{1} P_{2} = P$, $\tr(|P_{1} P_{2}|) = \tr(P) \leq \lfloor c\rfloor$, and so $(P_{1}, P_{2}) \in G_{r_{1}, r_{2}, c}$. \end{proof}

From now on, we take for granted that $r_{1} + r_{2} - d_{\mathcal{K}} \leq \lfloor c \rfloor$. Moreover, since $\tr(|P_{1} P_{2}|) \leq \min(\tr(P_{1}), \tr(P_{2}))$ for any two projectors $P_{1}$ and $ P_{2}$, we take $c \leq \min(r_{1}, r_{2})$.  Otherwise, the constraint on the product may as well not be considered.

The crucial ingredient of our analysis is a characterization of the angle between two subspaces due to Camille Jordan \cite{Jordan1875}. This characterization, commonly referred to as Jordan's lemma, allows us to reason about the relationship between the eigenvalues of a sum of two projectors and the singular values of their product. 
\begin{thm} 
\label{thm: two projector fan norm}
There exists $(P^{\text{opt}}_{1}, P^{\text{opt}}_{2}) \in G_{r_{1}, r_{2} ,c}$ that is maximal in the alignment ordering on $G_{r_{1}, r_{2} ,c}$, and it satisfies $\tr(|P^{\text{opt}}_{1} P^{\text{opt}}_{2}|) = c$. 
Moreover, if $c \in \mathbb{N}$, then $[P^{\text{opt}}_{1}, P^{\text{opt}}_{2}] = 0$.
\end{thm}

\begin{proof}
  It suffices to show that there exists $(P^{\text{opt}}_{1}, P^{\text{opt}}_{2}) \in G_{r_{1}, r_{2} ,c}$ such that 
  \begin{align}
      s_{1} P_{1}^{\text{opt}} + s_{2} P_{2}^{\text{opt}}  \succeq s_{1} P_{1} + s_{2} P_{2} 
  \end{align}
  holds for all $s_{1}, s_{2} \geq 0$ and $(P_{1}, P_{2}) \in G_{r_{1}, r_{2} ,c}$. We show existence by explicit construction. Let arbitrary $s_{1}, s_{2} \geq 0$ be given and consider an arbitrary pair $(P_{1},P_{2}) \in G_{r_{1}, r_{2},c}$. We modify the pair such that it remains in $G_{r_{1}, r_{2} ,c}$ and the corresponding Fan norms do not decrease until we reach an optimal pair.

We can decompose \(\mathcal{K}\) into minimal subspaces invariant under both \(P_{1}\) and \(P_{2}\). According to Jordan's lemma, these subspaces are either one- or two-dimensional. Thus, we may write
$\mathcal{K} = \mathcal{K}^{(1)} \oplus
\mathcal{K}^{(2)}$, where $\mathcal{K}^{(1)} = \bigoplus_{i_{1} = 1}^{m_{1}}
\mathcal{K}_{i_{1}}^{(1)}$ is the direct sum of invariant subspaces of dimension 1,
and $\mathcal{K}^{(2)} = \bigoplus_{i_{2} = 1}^{m_{2}}
\mathcal{K}_{i_{2}}^{(2)}$ is the direct sum of invariant subspaces of
dimension 2. Moreover, for each $i_{2} \in
[m_{2}]$,
$P_{1}|_{\mathcal{K}_{i_{2}}}$ and
$P_{2}|_{\mathcal{K}_{i_{2}}}$ are projectors of rank
$1$. Denote the overlap $\tr (| P_{1} P_{2} |) = \tilde{c}_{1} + \tilde{c}_{2}$ where $\tilde{c}_{1} = \tr(|P_{1} P_{2}|_{\mathcal{K}^{(1)}} |) \;
\text{and} \; \tilde{c}_{2} = \tr(|P_{1} P_{2}|_{\mathcal{K}^{(2)}} |)$. Observe that
$\tilde{c}_{1}$ is an integer as it equals the dimension of
$\text{supp}(P_{1}|_{\mathcal{K}^{(1)}}) \cap
\text{supp}(P_{2}|_{\mathcal{K}^{(1)}})$.

First, we reduce to the case where $m_{2} \in \{0, 1\}$. For each $i_{2} \in [m_{2}]$, we notate
\begin{align}
    P_{1}|_{\mathcal{K}_{i_{2}}^{(2)}} =: \ket{\alpha_{i_{2}}}\bra{\alpha_{i_{2}}} ,    P_{2}|_{\mathcal{K}_{i_{2}}^{(2)}} =: \ket{\beta_{i_{2}}}\bra{\beta_{i_{2}}}.
\end{align}
Notice that $\tilde{c}_{2} = \sum_{i_{2} = 1}^{m_{2}} | \braket{\alpha_{i_{2}}|\beta_{i_{2}}}|$. By Lem.~\ref{lemma: majorants}, the vector $(| \braket{\alpha_{i_{2}}|\beta_{i_{2}}}|)_{i_{2} = 1}^{m_{2}}$ is majorized by the vector $(\underbrace{1, \ldots, 1}_{\lfloor \tilde{c}_{2} \rfloor\; \text{times}}, \tilde{c}_{2} - \lfloor \tilde{c}_{2} \rfloor, 0, \ldots,0)$. For each $i_{2} \in [m_{2}]$, there are $\ket{\alpha_{i_{2}}'}, \ket{\beta_{i_{2}}'} \in \mathcal{K}_{i_{2}}^{(2)}$, such that 
\begin{align}
\label{eq: majorizing tuple}
(| \braket{\alpha_{i_{2}}'| \beta_{i_{2}}'}|)_{i_{2}=1}^{m_{2}} = (\underbrace{1, \ldots, 1}_{\lfloor \tilde{c}_{2} \rfloor\; \text{times}}, \tilde{c}_{2} - \lfloor \tilde{c}_{2} \rfloor, 0, \ldots,0).
\end{align}
 Specifically, for each $1$ on the right-hand side we choose the two vectors to be equal and for each $0$ we choose them to be orthogonal. For the remaining entry, we pick the two such that they have overlap $\tilde{c}_{2} - \lfloor \tilde{c}_{2} \rfloor$. By Lem.~\ref{lemma: isotone}, the eigenvalues of $(s_{1} P_{1} + s_{2} P_{2}) |_{\mathcal{K}^{(2)}}$ are a strictly isotone function of the nonzero singular values of $P_{1} P_{2}|_{\mathcal{K}^{(2)}}$, and so we may make the reassignments
\begin{align}
    \ket{\alpha_{i_{2}}} \gets \ket{\alpha_{i_{2}}'} \; , \;
    \ket{\beta_{i_{2}}} \gets \ket{\beta_{i_{2}}'}
\end{align}
for each $i_{2} \in [m_{2}]$. Save for at most one, this replaces each $\mathcal{K}_{i_{2}}^{(2)}$ with two one-dimensional invariant subspaces. That is, we may assume without loss of generality that $m_{2} \in \{0,1\}$ and that $\tilde{c}_{2} < 1$.

Second, we argue that $\tilde{c}_{1}$ may be taken to be at least $\lceil c \rceil - 1$ without loss of generality. Suppose that $\tilde{c}_{1} < \lceil c \rceil - 1$. Recall that the tuple of eigenvalues $\lambda((s_{1} P_{1} + s_{2} P_{2})|_{\mathcal{K}^{(1)}})$ contains $s_{1} + s_{2}$ with multiplicity $\tilde{c}_{1}$ along with the smaller eigenvalues $s_{1}, s_{2}$ and $0$. The number of one-dimensional invariant subspaces with eigenvalue \(s_{1}\) is at least $(r_1-1) - \tilde{c}_{1}$. Similarly, The number of one-dimensional invariant subspaces with eigenvalue \(s_2\) is at least $(r_2-1) - \tilde{c}_{1}$. The reason \(1\) is subtracted is to allow for the case where $m_2 = 1$. Since $\lceil c \rceil \leq \min(r_1, r_2)$, we may pair up $(\lceil c \rceil -1) - \tilde{c}_{1}$ many of these two kinds of invariant one-dimensional subspaces to get $(\lceil c \rceil -1) - \tilde{c}_{1}$ two-dimensional subspaces. We adjust the projectors so that they are equal when restricted to these subspaces. For each such subspace, this causes a transfer on the spectrum of $s_{1} P_{1} + s_{2} P_{2}$ of the form: 
\begin{align} 
 s_{1} \gets 0 \; , \; s_{2} \gets s_{1} + s_{2}.
\end{align}
 Since this is either an unjust transfer or an unjust transfer composed with a transposition, the Fan norms of $s_{1} P_{1} + s_{2} P_{2}$ do not decrease afterwards. Hence, we may, without loss of generality, take $\tilde{c}_1 \geq \lceil c \rceil - 1$. 

 Now, suppose $c$ is an integer. If $ m_2 = 1$, then we may adjust the action of the projectors on the sole invariant two-dimensional subspace so that they are equal when restricted to it. Then, $\tilde{c}_1$ becomes equal to $c$. If $m_{2} = 0$, then either $\tilde{c}_1 = c$ or there exists two invariant one-dimensional subspaces where the two projectors are nonzero and orthogonal. In the latter case, we may adjust the actions of the projectors on the direct sum of these two subspaces so that the two are equal there. In this case also $\tilde{c}_1$ becomes equal to $c$. In both cases, the arising pair
 $(P_{1}^{\text{opt}}, P_{2}^{\text{opt}})$ is a maximal pair in the alignment preorder. Notice that $\tr (|P_{1}^{\text{opt}} P_{2}^{\text{opt}}|) = c$ and $[P_{1}^{\text{opt}}, P_{2}^{\text{opt}}] = 0$.

 On the other hand, consider the case where $c$ is not an integer. Then, $\lceil c \rceil - 1 = \lfloor c \rfloor$. If $m_{2} = 1$, then we adjust the overlap in the two-dimensional invariant subspace until $\tilde{c}_{2} = c - \lfloor c \rfloor$.  That is, we make the projectors align there as much as possible. From Eq.~\ref{eq: eigenvalue equation for two rank1 projector sum}, this is either an unjust transfer or an unjust transfer composed with a transposition on the spectrum of \(s_1 P_1 + s_2 P_2\). If $m_{2} = 0$, then consider again two invariant one-dimensional subspaces where the two projectors are nonzero and orthogonal. 
 Such a pair exists because $\tilde{c}_1 = \lfloor c \rfloor < c \leq \min(r_1, r_2)$. Then, we adjust the action of the two projectors on the direct sum of these two subspaces so that $m_{2} = 1$ and $\tilde{c}_{2} = c -\lfloor c \rfloor$. Again, the pair arising after these adjustments $(P^{\text{opt}}_{1}, P^{\text{opt}}_{2})$ is a desired maximal pair.
\end{proof}

By Thm.~\ref{thm: two projector fan norm}, if $c \in \mathbb{N}$, then any pair of commuting projectors in $G_{r_{1}, r_{2} ,c}$ whose product has rank $c$ is optimal. By Lem.~\ref{overlap lemma}, the projectors in the alignment operator in Eq.~\ref{spin alignment with mu 2} have overlap at most $(\text{rank}(Q))^{|(I_{1} \cup I_{2})^{c}|}$.
\begin{cor} 
\label{corollary: two term projector fan norms optimality}
   For $t \in [0,1], I_{1}, I_{2} \subseteq [n]$ and $Q$ proportional to a projector,
   \begin{align*}
   \label{rel: alignment for two terms.}
   (1 - t) \ket{q_{1}}\bra{q_{1}}^{\otimes {I_{1}}} \otimes Q^{\otimes I_{1}^{c}} + t \ket{q_{1}}\bra{q_{1}}^{\otimes {I_{2}}} \otimes Q^{\otimes I_{2}^{c}} \succeq
   (1 - t) \ket{\psi_{{1}}}\bra{\psi_{{1}}}_{I_{1}} \otimes Q^{\otimes I_{1}^{c}} + t \ket{\psi_{{2}}}\bra{\psi_{{2}}}_{I_{2}}  \otimes Q^{\otimes I_{2}^{c}}.
   \end{align*}  
\end{cor}

\section{``Dual" problems to spin alignment and their generalizations}
\label{sec:``dual" problem}
We introduce a general class of conjectures that articulates the intuition that under a fixed global spectrum constraint, the least locally dispersed states are classical. The inspiration came from considering spin alignment problems with $Q \propto \mathds{1}$ and using Fan's maximum principle Eq.~\ref{eq: Fan maximum prinicple}. Specifically, for a positive measure $\nu$ and a state tuple $(\ket{\psi_{I}} \bra{\psi_{I}})_{I \subseteq [n]}$, we consider 

\begin{align}
    \max_{P: \text{rank}(P) = k} \tr( P_{[n]} \sum_{I \subseteq [n]} \nu_{I} \ket{\psi_{I}} \bra{\psi_{I}} \otimes \mathds{1}^{\otimes I^{c}}) &=    \max_{P: \text{rank}(P) = k} \sum_{I \subseteq [n]}  \nu_{I} \tr( P_{I} \ket{\psi_{I}} \bra{\psi_{I}}) \\ &\leq    \max_{P: \text{rank}(P) = k} \sum_{I \subseteq [n]} \nu_{I} \lambda_{1} (P_{I}).   
\end{align}

 When $(\ket{\psi_{I}} \bra{\psi_{I}})_{I \subseteq [n]} = (\ket{q_{1}}\bra{q_{1}}^{\otimes I})_{I \subseteq [n]}$, the inequality is satisfied with equality. This maximization problem may be
generalized by letting the optimization be over quantum states
with a prescribed global spectrum, and by replacing
$\lambda_{1}$ with any continuous, convex,
unitarily invariant function. Moreover, since the objective function is convex,
we may take the convex closure of the feasible set without loss of
generality. We formalize this below.

Let complex inner product spaces $\mathcal{H}_{1}, ..., \mathcal{H}_{n}$ be given and denote $\mathcal{H}_{I} = \bigotimes_{i \in I} \mathcal{H}_{i}$ for $I \subseteq [n]$. For a probability measure $p \in \mathbb{R}^{\Pi_{i \in [n]} d_{\mathcal{H}_{i}}}$, we consider the maximization problem:
\begin{align} 
\label{``dual" program conjecture}
    &\max_{\tau_{[n]} \in \mathcal{D} (\mathcal{H}_{[n]})} \quad \quad \sum_{I \subseteq [n]} f_{I} ( \tau_{I}), \\
    &\text{subject to:} \quad \quad \; \lambda( \tau_{[n]} ) \preceq p, 
\end{align}
where for all $I \subseteq [n]$, $f_{I}$ is continuous, convex and unitarily invariant on a domain containing $\mathcal{D} (\mathcal{H}_{I})$. We conjecture that there exists a classical optimal point. This is self-evident in the case where $p$ has only one nonzero value, as all the functions in the sum can be maximized simultaneously. More interesting are the cases where the optimization is over strictly mixed states. An example of such a problem in the quantum information literature appears in Ref.~\cite{Wilde2014}. Namely, Eq.~39 there gives an upper bound on the fidelity of entanglement generating pure codes in terms of a sum as above where for each $I \subseteq [n]$, $f_{I} (\cdot) \propto || \cdot ||_{\alpha}^{\alpha}$ for a fixed $\alpha \in [1,2]$.

Below, we show that this ``dual" conjecture holds in the special case where non-constant functions in $(f_{I})_{I \subseteq [n]}$ correspond to members of a partition of $[n]$. 
\begin{prop} 
\label{``dual" program with part}
    Let $\mathcal{W} \subseteq 2^{[n]}$ be a partition of $[n]$. The problem
    \begin{align} 
    \label{``dual" program statement for partitioning}
        &\max_{\tau_{[n]} \in \mathcal{D} (\mathcal{H}_{[n]})} \quad \quad \sum_{I \in \mathcal{W}} f_{I} ( \tau_{I}), \\
        &\text{subject to:} \quad \; \;  \lambda( \tau_{[n]} ) \preceq p. 
    \end{align}
    has a classical solution. 
\end{prop}
\begin{proof}
    Without loss of generality, let $\tau_{[n]}\in \mathcal{D} (\mathcal{H}_{[n]})$ have spectrum $p$. For each $I \in \mathcal{W}$, if necessary, a unitary quantum channel $\mathcal{U}_{I} \otimes \id_{I^{c}}$ may be applied so that $\tau_{I}$ is classical. This does not affect the objective function. For each $I \in \mathcal{W}$, let $\mathcal{N}_{\tau, I}$ denote the quantum channel that fully decoheres in a basis of $\tau_{I}$. That is, it has the pinching action $ \mathcal{N}_{\tau, I} (\cdot ) = \text{diag}(\cdot)$. It is clear that $\tau'_{[n]} = \bigotimes_{I \in \mathcal{W}} \mathcal{N}_{\tau, I} (\tau_{[n]})$ is classical, and for each $I \in \mathcal{W}$, $\tau'_{I} = \tau_{I}$. Since pinching is mixed-unitary (see chapter 4 of Ref.~\cite{Watrous2018}
    ), then $\lambda(\tau'_{[n]} ) \preceq  p$. We may optimize further by noticing that \(\tau'_{[n]}\) is a convex combination of classical states each with spectrum equal to $p$. One of them has a higher objective value. 
\end{proof}

\section{Concluding remarks}
\label{sec:conclusion}
We generalized and systematically studied spin alignment problems. We gave non-trivial examples of cases where tuples of spectrally constrained self-adjoint operators are more aligned than others. These can be considered refinements of Ky Fan's relation Eq.~\ref{Fan's majorization relation}. We conclude with a non-exhaustive list of problems in this area that we hope will help guide future research.
\begin{enumerate}
    \item (\textbf{Schatten norms of non-integer order}) We used Cor.~\ref{cor: ovelap of absolute value of trace} to resolve the spin alignment problem for Schatten norms of integer order. It is of interest to know if a similar approach can be used in cases where the order of the norm is not an integer. Suppose that $A_{0}, A_{1}, B_{0}, B_{1} \geq 0$ satisfy $\lambda(A_{0}) = \lambda(B_{0})$, $\lambda(A_{1}) = \lambda(B_{1})$. For a string $s \in \{0, 1\}^{*}$, define the ordered products $\Pi_s (A_{0}, A_{1}) := \Pi_{i \in [|s|]} A_{s(i)}$ and $\Pi_s(B_{0}, B_{1}) := \Pi_{i \in [|s|]} B_{s(i)}$. If for all strings $s$, $\tr(\Pi_s (A_{0}, A_{1})) \geq |\tr(\Pi_s (B_{0}, B_{1}))|$, does it hold that $|| A_{0} + A_{1} ||_{p} \geq || B_{0} + B_{1} ||_{p}$ for all $p \in [1,\infty)$? If so, then the overlap lemma may be used to prove the spin alignment conjecture for all Schatten norms, and by extension for the von Neumann entropy as well. 

    \item (\textbf{Algebraic approximations of Fan norms}) The Schatten norms of integer order may be said to be algebraic norms, in that they can be defined using polynomials of the absolute value of the operator in question. Since $\lim_{m \rightarrow \infty} || \cdot ||_{m} = || \cdot ||$, they approximate the Fan norm of order $1$ arbitrarily well. Clearly, the trace norm is algebraic as well. Do there exist algebraic operator norms that approximate other Fan norms? If so, is there any interesting structure to the corresponding polynomials, especially with regard to spin alignment?

    \item (\textbf{Separable refinement of Ky Fan's relation}) Given $A_{0}, A_{1}, B_{0}, B_{1} \geq 0$, does the following majorization relation
    \begin{align}
    \label{rel: sep Ky Fan}
    \lambda ( A_{0} \otimes B_{0} + A_{1} \otimes B_{1}) \preceq  \lambda ( A^{\downarrow}_{0} \otimes B^{\downarrow}_{0} + A^{\downarrow}_{1} \otimes B^{\downarrow}_{1})    
    \end{align}
    hold? In the special case where $\text{rank} (B_{0}) = \text{rank} (B_{1}) = 1$, we know the answer to be yes by way of a characterization of separable states due to Nielsen and Kempe \cite{Nielsen2001}. Specifically, they showed that separable states contain more dispersion globally than locally. If the relation Eq.~\ref{rel: sep Ky Fan} holds more generally, then it may be used to prove a separable version of the strong spin alignment conjecture.   
\end{enumerate}

\subsubsection*{Acknowledgements}
The authors thank Eli Halperin, Graeme Smith, and Mark Wilde for fruitful discussions regarding the manuscript, and Felix Leditzky and Scott Glancy for comments that helped improve the presentation of the material. The authors also thank the anonymous referee for their careful reading of the manuscript and useful comments. At the time this work was performed, M. A. Alhejji was supported as an Associate in the Professional Research
Experience Program (PREP) operated jointly by NIST
and the University of Colorado Boulder. This is a contribution of the National Institute of
Standards and Technology, not subject to U.S. copyright.

\subsubsection*{Declarations}
\textbf{Funding and/or Conflicts of interests/Competing interests}: we declare no conflict of interests. \textbf{Data availability}: the manuscript has no associated data.

\appendix
\section{The overlap lemma}
\label{sec:app A}

\begin{lem} 
\label{overlap lemma} (The overlap lemma)
Let $I_{1},\ldots,I_{\ell}$ be a family of subsets of $[n]$. For \(i \in [\ell]\), let \(Q_{I_{i}^{c}} = \bigotimes_{j \in I_{i}^{c}} Q_{I_{i}^{c},j}\), where for each \(j \in I_{i}^{c}\), \(Q_{I_{i}^{c},j} \in \mathcal{B}(\mathcal{H}_{j})\) has largest singular value \(\alpha_{i,j} > 0\) and is of the form \(Q_{I_{i}^{c},j} = \alpha_{i, j} \ket{q_1}\bra{q_{1}} \oplus 
W_{i, j}\). For any unitarily invariant norm $||| \cdot |||$, the maximization problem 
\begin{align}  \label{overlap equation}  \max_{(R_{I_{i}})_{i=1}^{\ell}}     ||| \; \;\prod_{i=1}^{\ell} R_{I_{i}} \otimes Q_{I_{i}^{c}} \; |||,
\end{align}
where the variable tuple $(R_{I_{i}})_{i=1}^{\ell}$ is of operators each with trace norm at most \(1\), has $(\ket{q_{1}}\bra{q_{1}}^{\otimes I_{i}})_{i=1}^{\ell}$ as a solution. \end{lem}

We argue by induction on the number of sets \(\ell\). For that, we need the following statement. 
\begin{lem}
\label{lem: inductive step for the overlap lemma proof}
Let \(\cK_1, \cK_2,\) and \(\cK_3\) be given complex inner product spaces. Let \(A_{1} \in \cB(\cK_1)\) and \(A_{3} \in \cB(\cK_3)\) be such that \( ||A_{1}|| \leq 1\) and \(||A_{3}|| \leq 1\). For all \(T_{2,3} \in \cB(\cK_2 \otimes \cK_3)\)  and \(S_{1,2} \in \cB(\cK_1 \otimes \cK_2)\) of trace norm at most \(1\), it holds that \(
    || (A_{1} \otimes T_{2,3}) (S_{1,2} \otimes A_{3})||_{1} \leq 1\).
\end{lem}
\begin{proof}
Because \(|| \cdot ||_1\) is homogeneous and convex, it suffices to prove the inequality holds in the cases where \(S_{1,2} = \ket{s}\bra{\tilde{s}}_{1,2}\) and \(T_{2,3} = \ket{t}\bra{\tilde{t}}_{2,3}\) for arbitrary unit vectors \(\ket{s}, \ket{\tilde{s}} \in \cK_1 \otimes \cK_2\) and \(\ket{t}, \ket{\tilde{t}} \in \cK_{2}\otimes \cK_3\). Since \(||A_{1}|| \leq 1\) and \(||A_{3}|| \leq 1\), the norms of \(\ket{v}_{1,2} = A_{1} \otimes \one_2 \ket{s}_{1,2}\) and \(\ket{u}_{2,3} = \one_2 \otimes A_{3}^{*} \ket{\tilde{t}}_{2,3}\) are bounded from above by \(1\). Hence, it suffices to prove that 
\begin{align}
    || (\one_1 \otimes \ket{t}\bra{u}_{2,3}) (\ket{v}\bra{\tilde{s}}_{1,2} \otimes \one_3||_1 \leq 1
\end{align}
for arbitrary unit vectors \(\ket{v}, \ket{\tilde{s}} \in \cK_1 \otimes \cK_2\) and \(\ket{t}, \ket{u} \in \cK_{2}\otimes \cK_3\). 

Consider the inner factor operator \(\bra{u}_{2,3} \ket{v}_{1,2} : \cK_3 \rightarrow \cK_1\). To see that its trace norm is at most 1, let \(\ket{v}_{1,2} = \sum_{i} \sqrt{\lambda_i} \ket{\alpha_{i}}_{1} \otimes \ket{\beta_{i}}_{2}\) and \(\ket{u}_{2,3} = \sum_{j} \sqrt{\mu_j} \ket{\gamma_{j}}_{2} \otimes \ket{\delta_{j}}_{3}\) be Schmidt decompositions. Then,
\begin{align}
    \bra{u}_{2,3} \ket{v}_{1,2} = \sum_{i, j} \sqrt{\mu_j} \sqrt{\lambda_{i}} \braket{\gamma_{j}|\beta_{i}} \ket{\alpha_{i}}_{1}\bra{\delta_{j}}_{3}
\end{align}
and 
\begin{align}
    |\bra{u}_{2,3} \ket{v}_{1,2}|^{2} = \sum_{i, j, j'} \sqrt{\mu_{j} \mu_{j'}} \lambda_{i} \braket{\gamma_{j}|\beta_{i}} \braket{\beta_{i}|\gamma_{j'}} \ket{\delta_{j'}}\bra{\delta_{j}}_{3}.
\end{align}
Denote \( b_j = \bra{\gamma_j} (\sum_{i} \lambda_{i} \ket{\beta_i}\bra{\beta_{i}}) \ket{\gamma_{j}} \) and notice that \(\sum_{j} b_{j} \leq 1\). It follows from the Schur-concavity of \(\tr(\sqrt{\cdot})\) and the Schur-Horn theorem that 
\begin{align}
    || \bra{u}_{2,3} \ket{v}_{1,2} ||_{1} &= \tr ( \sqrt{|\bra{u}_{2,3} \ket{v}_{1,2}|^{2}}) \\
    &\leq \tr (\sqrt{\text{diag}(|\bra{u}_{2,3} \ket{v}_{1,2}|^{2})}) \\
    &= \sum_{j} \sqrt{\mu_j b_j} \leq 1.
\end{align}
Finally, \(
||  \bra{u}_{2,3} \ket{v}_{1,2} \otimes \ket{t}_{2,3}\bra{\tilde{s}}_{1,2} ||_1 = || \bra{u}_{2,3} \ket{v}_{1,2} ||_1 \; ||\ket{t}_{2,3}\bra{\tilde{s}}_{1,2} ||_1 \leq 1. \) \end{proof}

\begin{proof} [Proof of Lem.~\ref{overlap lemma}]
  Since \(||| \cdot |||\) is homogeneous, it may
  be assumed without loss of generality for each \(i \in [\ell]\) that \(||Q_{I_{i}^{c}, j} || = 1\) for all \(j \in I_{i}^{c}\). We proceed via induction on \(\ell\). If \(\ell = 1\), the statement of the lemma follows from the facts that \(|| \ket{q_1}\bra{q_1}^{\otimes I_1}||_1 = 1\) and that \(||| \cdot |||\) is convex and unitarily invariant. Suppose that the statement holds in cases where the family of sets has \(\ell > 1\) elements. Consider \(\prod_{i=1}^{\ell +1} R_{I_{i}} \otimes Q_{I_{i}^{c}}\) and notice that
    \begin{align}
    \label{eq: two operator product}
         R_{I_{1}} \otimes Q_{I_{1}^{c}} R_{I_{2}} \otimes Q_{I_{2}^{c}} = \underbrace{Q_{I_{1}^{c} \cap I_2} R_{I_{1}} R_{I_{2}} Q_{I_{2}^{c} \cap I_{1}} }_{ \tilde{R}_{ I_{1} \cup I_{2} }} \otimes  \; \underbrace{Q_{I_{1}^{c} \setminus I_2}  Q_{I_{2}^{c} \setminus I_1}}_{\tilde{Q}_{ (I_{1} \cup I_{2})^{c}}}.
    \end{align}
    The operator \(\tilde{Q}_{ (I_{1} \cup I_{2})^{c}}\) is completely factorizable on \(I_{1} \cup I_{2}\) and has maximal singular value \(1\). Moreover, each of its factors can be written as \(\ket{q_{1}}\bra{q_1} \oplus W\) for some operator \(W\). By Lem.~\ref{lem: inductive step for the overlap lemma proof}, \(|| \tilde{R}_{ I_{1} \cup I_{2}} ||_{1} \leq 1\) and so we may estimate
    \begin{align}
       ||| \prod_{i=1}^{\ell +1} R_{I_{i}} \otimes Q_{I_{i}^{c}} ||| &= |||  \tilde{R}_{ I_{1} \cup I_{2} } \otimes  \; \tilde{Q}_{ (I_{1} \cup I_{2})^{c}} \prod_{i=3}^{\ell + 1} R_{I_{i}} \otimes Q_{I_{i}^{c}}||| \\
        &\leq ||| \ket{q_{1}}\bra{q_1}^{\otimes I_{1} \cup I_{2}} \otimes \; \tilde{Q}_{ (I_{1} \cup I_{2})^{c}} \prod_{i=3}^{\ell + 1} \ket{q_1}\bra{q_1}^{\otimes I_{i}} \otimes Q_{I_{i}^{c}} ||| \\
        &=  ||| \prod_{i=1}^{\ell +1} \ket{q_{1}}\bra{q_{1}}^{\otimes I_{i}} \otimes Q_{I_{i}^{c}} |||, 
    \end{align}
    where the inequality is by the induction hypothesis.  \end{proof}

\section{On the relationship between the sum of two projectors and their product}
\label{sec:app B}

In this appendix, we prove lemmas necessary to elucidate the relationship between a non-negative linear combination $s_{1} P_{1} + s_{2} P_{2}$ of two projectors $P_{1}, P_{2} 
\in \mathcal{S}(\mathcal{K})$ and their product $P_{1} P_{2}$. By \cite{Jordan1875}, $\mathcal{K}$ may be decomposed into a direct sum of subspaces, each of dimension at most $2$, that are invariant under the action of both $P_{1}$ and $P_{2}$. Moreover, when restricted to each invariant subspace, the two projectors have rank at most 1. So, we may write $\mathcal{K}$ as an orthogonal direct sum of one- and two-dimensional minimal invariant subspaces
\begin{align}
\mathcal{K} = \bigoplus_{i_{1} = 1}^{m_{1}} 
\mathcal{K}_{i_{1}}^{(1)} \oplus \bigoplus_{i_{2} = 1}^{m_{2}} \mathcal{K}_{i_{2}}^{(2)}, 
\end{align}
where the $\mathcal{K}_{i}^{(1)}$ are one-dimensional and the $\mathcal{K}_{i}^{(2)}$ are two dimensional. By minimal, we mean that the subspaces contain no proper nonzero invariant subspace.
For each $i_{2} \in [m_{2}]$, we notate
\begin{align}
    P_{1}|_{\mathcal{K}_{i_{2}}^{(2)}} =: \ket{\alpha_{i_{2}}}\bra{\alpha_{i_{2}}} \; , \;    P_{2}|_{\mathcal{K}_{i_{2}}^{(2)}} =: \ket{\beta_{i_{2}}}\bra{\beta_{i_{2}}}.
\end{align}
$P_{1}$ and $P_{2}$ commute if and only if $m_{2} = 0$. The difficulty in reasoning about the eigenvalues of a linear combination of two projectors lies in these subspaces where they do not commute.

When restricted to $\bigoplus_{i_{2} = 1}^{m_{2}} \mathcal{K}_{i_{2}}^{(2)}$, the nonzero singular values of $P_{1} P_{2}$ are $(| \braket{\alpha_{i_{2}}|\beta_{i_{2}}}|)_{i_{2} = 1}^{m_{2}}$. The eigenvalues of the restriction $(s_{1} P_{1} + s_{2} P_{2}) |_{\mathcal{K}_{i_{2}}^{(2)}}$ may be computed as
\begin{align}
\label{eq: eigenvalue equation for two rank1 projector sum}
\frac{1}{2} ( (s_{1} + s_{2}) \pm \sqrt{(s_{1} - s_{2})^{2} + 4 s_{1} s_{2} | \braket{\alpha_{i_{2}}|\beta_{i_{2}}}|^{2}}).
\end{align}
Hence, the eigenvalues of $s_{1} P_{1} + s_{2} P_{2}$ are a function of the singular values of the product $P_{1} P_2$. We show next that it is in fact a \textit{strictly isotone} function (see page 41 of Ref.~\cite{Bhatia1997}). A strictly isotone function $G$ is one that preserves the majorization ordering in the sense that if $v\succeq w$, then $G(v)\succeq G(w)$.  That is, the less dispersed the singular values of $P_{1} P_{2}$, the less dispersed the eigenvalues of $s_{1} P_{1} + s_{2} P_{2}$. This is a consequence of the fact that for $a , b \geq 0$, the map $x \mapsto \sqrt{a + b x^{2}}$ is convex on $\mathbb{R}_{\geq 0}$. 
\begin{lem} \label{lemma: isotone}
Let $A \subseteq \mathbb{R}$ be convex and $g: A \rightarrow \mathbb{R}_{\geq 0}$ be a convex function. For $t \in \mathbb{R}$, define the mapping $G: A^{m} \rightarrow \mathbb{R}^{2 m}$ with action
\begin{align}
(v_{1},\ldots, v_{m}) \mapsto (t + g(v_{1}), \ldots, t + g(v_{m}) ) \oplus (t - g(v_{1}), \ldots, t - g(v_{m}) ).
\end{align}
Then, $G$ is strictly isotone. \end{lem}
\begin{proof}
Let $v, w \in A^{m}$ be such that $v \succeq w$. Observe that for $k \in [m]$, 
\begin{align}
\sum_{j=1}^{k} G(v)^{\downarrow}_{j} = k t + \sum_{j=1}^{k} g(v)_{j}^{\downarrow} \geq k t + \sum_{j=1}^{k} g(w)_{j}^{\downarrow} = \sum_{j=1}^{k} G(w)^{\downarrow}_{j},
\end{align}
where the inequality follows from the convexity of $g$ and the doubly-stochastic characterization of majorization (see, for example, Theorem II.3.3 on page 41 of Ref.~\cite{Bhatia1997}).  If $k > m$, then 
\begin{align}
\sum_{j=1}^{k} G(v)^{\downarrow}_{j} &= k t + \sum_{j=1}^{m} g(v)_{j}^{\downarrow} - \sum_{j=1}^{k - m} g(v)_{j}^{\uparrow} = k t + \sum_{j=1}^{2m - k} g(v)_{j}^{\downarrow}\\
&\geq k t + \sum_{j=1}^{2m - k} g(w)_{j}^{\downarrow} = \sum_{j=1}^{k} G(w)^{\downarrow}_{j}.
\end{align}
Since $\sum_{j=1}^{2 m} G(\cdot)_{j} = 2 m t$, $G(v) \succeq G(w)$. \end{proof}

The following lemma exhibits the maximal elements in the majorization ordering in a superset of the possible tuples of nonzero singular values of $P_{1} P_{2}|_{\bigoplus_{i_{2} = 1}^{m_{2}} \mathcal{K}_{i_{2}}^{(2)}}$. 
\begin{lem} \label{lemma: majorants}
Given $m \in \mathbb{N}, e \geq 0$, define the set
\begin{align}
    S_{e} := \{ x \in \mathbb{R}^{m} \; | \; \forall i \in [m], x_{i} \in [0,1], \sum_{i=1}^{m} x_{i} = e\}.
\end{align}
The element $(\underbrace{1, \ldots, 1}_{\lfloor e \rfloor\; \text{times}}, e - \lfloor e \rfloor, 0, \ldots,0)$ is a majorant of $S_{e}$. 
\end{lem}
\begin{proof}
Let $x \in S_{e}$ be arbitrary. For $k \in [m]$, $k \leq \lfloor e \rfloor$, observe that 
$\sum_{i = 1}^{k} x^{\downarrow}_{i} \leq \sum_{i=1}^{k} 1 = k$. And for $k > \lfloor e \rfloor$, $\sum_{i = 1}^{k} x^{\downarrow}_{i} \leq \sum_{i = 1}^{m} x^{\downarrow}_{i} = e$. 
\end{proof}

\section{Proof of the conjecture in Ex.~\ref{ex: toy example}}
\label{sec:app C}
Let $\sum_{i = 1}^{d-1} \lambda_i(\tau) \ket{\tau_i}\bra{\tau_i}$ be a spectral decomposition of $\tau$. Let arbitrary $p \in [0, 1]$ be given and consider the convex mixture $\tau_{\gamma} := p \tau + (1-p) \ket{v_{\gamma}}\bra{v_{\gamma}}$. We wish to prove that $\tau_\gamma \preceq \tau_1$. Define
\begin{align}
    w :=  (\sqrt{p(1-p)\lambda_1 (\tau)} \braket{\tau_1 | \alpha},\ldots, \sqrt{p(1-p)\lambda_{d - 1} (\tau) } \braket{\tau_{d-1} | \alpha})^T. 
\end{align}
The Gram matrix for the \(d\) vectors \(\sqrt{p\lambda_{1}}\ket{\tau_{1}},
  \ldots \sqrt{p\lambda_{d-1}}\ket{\tau_{d-1}}, \sqrt{1-p}\ket{v_{\gamma}}\) is
\begin{align}
\label{eq: Gram matrix form}
M_{\gamma} :=  
\begin{pmatrix} 
  \begin{matrix}
    p \lambda_{1} (\tau) & \cdots & 0 \\
    \vdots& \ddots & \vdots\\
    0 & \cdots & p \lambda_{d - 1} (\tau)
  \end{matrix} \;
  \vline
  &  \begin{matrix}
          \sqrt{\gamma} w
         \end{matrix} \\
\hline
 \begin{matrix}
          \quad \quad \quad  \quad \sqrt{\gamma} w^*
         \end{matrix}  \hfill
    \vline
    &
    (1 - p)
\end{pmatrix}
\end{align}
It is not difficult to show that (see \cite{Fannes_2012} for example), up to zeros, the spectrum of $M_{\gamma}$ is equal to the spectrum of $\tau_{\gamma}$. Hence, it suffices to show that $M_{\gamma} \preceq M_{1}$ for all $\gamma \in [0,1]$.
This statement is proven
  in the next paragraph.

Let \(\ket{e_1}, \ldots, \ket{e_d}\) denote the orthonormal basis (ordered in the obvious way) used to write down the Gram matrices in Eq.~\ref{eq: Gram matrix form}. Consider the unitary 
\begin{align}
    V := (\sum_{i = 1}^{d-1} \ket{e_i}\bra{e_i}) - \ket{e_d}\bra{e_d}. 
\end{align}
Supposing it occurs with a probability \(q \in [0,1]\), we have
\begin{align}
    (1-q) M_1 + q V M_1 V^* = \begin{pmatrix} 
  \begin{matrix}
    p \lambda_{1} (\tau) & \cdots & 0 \\
    \vdots& \ddots & \vdots\\
    0 & \cdots & p \lambda_{d - 1} (\tau)
  \end{matrix} \;
  \vline
  &  \begin{matrix}
          (1-2q) w
         \end{matrix} \\
\hline
 \begin{matrix}
          \quad  \quad (1-2q) w^*
         \end{matrix}  \hfill
    \vline
    &
    (1 - p)
\end{pmatrix}
\end{align}
If \(q = \frac{1 - \sqrt{\gamma}}{2}\), then \((1-q) M_1 + q V M_1 V^* = M_{\gamma}\). Hence, for all \(\gamma \in [0,1]\), there exists a mixed-unitary channel that takes \(M_1\) to \(M_\gamma\).

\bibliographystyle{amsplain}
\bibliography{references}
\end{document}